\documentclass[aps,10pt,pre,nofootinbib]{revtex4-2}
\makeatletter
\let\auto@bib@innerbib\@empty
\let\write@bibliographystyle\relax
\makeatother
\usepackage{amsmath,amssymb,amsthm}
\usepackage{graphicx}

\usepackage{xcolor}
\newtheorem{theorem}{Theorem}
\newtheorem{proposition}{Proposition}
\newtheorem{lemma}{Lemma}
\newtheorem{corollary}{Corollary}
\newtheorem{remark}{Remark}
\theoremstyle{definition}
\newtheorem{assumption}{Assumption}
\newcommand{\eps}{\varepsilon}
\newcommand{\tho}{\theta_0}
\newcommand{\rhoo}{\rho_0}
\newcommand{\Th}{\Theta}
\newcommand{\Si}{\Sigma}
\newcommand{\bmuo}{\bar{\mu}_0}
\newcommand{\Mst}{M_0^{\ast}}

\newcommand{\calN}{\mathcal{N}}
\newcommand{\calE}{\mathcal{E}}

\newcommand{\dMach}{d_M}
\newcommand{\dDef}{d_F}
\begin{document}
\title{Nonisothermal Shock Structure and Universal Flow-Index Thresholds \\ in a Hyperbolic Power-Law Fluid}
\author{Tommaso Ruggeri}
\email{tommaso.ruggeri@unibo.it}
\affiliation{Department of Mathematics and AM$^2$C Research Center, University of Bologna, Via Saragozza 8, 40123 Bologna, Italy}
\affiliation{Accademia Nazionale dei Lincei, Via della Lungara 10, 00165 Roma, Italy}
\date{\today}
\begin{abstract}
We investigate whether the two flow-index thresholds previously found for isothermal shock profiles persist when the full nonisothermal dynamics is taken into account in a hyperbolic power-law relaxation model of Rational Extended Thermodynamics. The nonisothermal profile problem is structurally different from its isothermal counterpart: restoring the energy balance determines the temperature along the traveling wave and feeds it back into the pressure, the relaxation production, the temperature-dependent consistency coefficient, and the characteristic structure. Under thermodynamic stability, $p_\theta\ge0$, and strict convexity of the reduced Hugoniot pressure, no nontrivial constant-temperature compressive profile can satisfy the full equations. We derive an exact global characteristic-ordering identity and prove that the positive nonequilibrium characteristic speed has its strict global minimum at the unperturbed upstream state. Consequently, a monotone continuous profile exists for $1<M_0<\Mst$, whereas for $M_0>\Mst$ the Boillat--Ruggeri theorem excludes a $C^1$ profile and any admissible piecewise-smooth connection must contain a subshock. Despite the thermomechanical coupling, the shock-thickness classification remains unchanged: $m=2$ is the weak-shock threshold and $m=1$ the near-critical threshold as $M_0\nearrow\Mst$. The corresponding exponents are constitutive-independent within the present class, while finite limiting values and prefactors depend on the equation of state, internal energy, and temperature-dependent consistency coefficient. For the Tait--Murnaghan example, increasing the reference temperature lowers the critical Mach number when the dimensional viscous--relaxation scale is fixed and, for the thermally thinning law considered, reduces the resolved shock thickness.
\end{abstract}
\maketitle 
\section{Introduction}
\subsection{Background and motivation: Shock waves in soft matter}
Complex fluids and soft matter systems---encompassing polymer solutions, biological gels, and dense colloidal suspensions---exhibit highly nonlinear, rate-dependent responses that defy classical Newtonian hydrodynamics.
	Their defining macroscopic features, such as shear thinning and shear thickening, emerge from complex mesoscopic structural rearrangements under applied stress \cite{Bird1987, ChhabraRichardson2008, BrownJaeger2014}.
	While the steady-state rheology of these materials is extensively documented, their behavior under extreme dynamical conditions, such as the propagation of shock waves and high-rate pressure pulses, remains a profound challenge in nonlinear physics.
	Understanding shock structures in complex fluids is crucial for phenomena ranging from ballistic impact mitigation in biological tissues to the propagation of transient fronts in dense suspensions \cite{Waitukaitis2012, Peters2016}. These phenomena motivate the study, but the present macroscopic RET closure does not attempt to model microstructural jamming mechanisms explicitly.
	In these high-strain-rate regimes, a widely used constitutive relation is the power-law (or Ostwald--de Waele) rheology:
\begin{equation}\label{powerlawintro}
		\mathbf{s} = k(\theta)\,\dot\gamma^{m-1}\,\mathbf{D},
		\qquad \dot\gamma = 2\|\mathbf{D}\|,
	\end{equation}
    where $\mathbf{s}$ denotes the viscous, nonequilibrium part of the Cauchy stress tensor,
	 $m>0$ is the flow-behavior index, $k(\theta)>0$ is the temperature-dependent consistency coefficient, and $\mathbf{D}$ is the strain-rate tensor. 
     In the Lagrangian formulation used below, the corresponding viscous first Piola--Kirchhoff stress is denoted by $\boldsymbol{\sigma}$ and is related to the Cauchy viscous stress by the standard Piola transformation
$\boldsymbol{\sigma}=\mathbf{s}\,\mathbf{F}^{C}$, where $\mathbf{F}^{C}$ is the cofactor of the deformation gradient. In the one-dimensional reduction this becomes the scalar variable $\sigma$.
	The temperature sensitivity of $k(\theta)$ can be essential: viscous dissipation within a shock layer may generate appreciable heating and thereby alter the local mobility of the complex fluid.
	This is not only a modeling refinement, but a central issue in the rheology of complex and polymeric fluids.
	Temperature changes affect relaxation times and material functions, as expressed for instance by the Williams--Landel--Ferry time--temperature superposition principle \cite{WLF1955} and by the classical theory of polymeric liquids \cite{Bird1987}.
	The development of constitutive equations for nonisothermal polymeric flows, including thermal histories, viscous heating, and the coupling between the energy equation and the stress response, was emphasized in the review by Bird and Wiest \cite{BirdWiest1995}.
	Non-isothermal channel flows of non-Newtonian fluids with viscous heating were studied by Dinh and Armstrong \cite{DinhArmstrong1982}.
	Thermo-rheological models in which the viscosity, consistency coefficient, or even the power-law exponent depend on temperature have also been analyzed in the mathematical literature \cite{AntontsevRodrigues2006,GrasselliParoliniPoiattiVerani2023}.
	Related effects of viscous heating in highly viscous fluids with temperature-dependent viscosity have been discussed, for example, in magma-flow models \cite{CostaMacedonio2003}.
	These works show that temperature is not a secondary correction in non-Newtonian fluids: it may modify the effective viscosity, the relaxation time scale, and the dissipation mechanism.
	Any realistic physical analysis of shock transitions must therefore account for the full thermomechanical coupling introduced by the energy equation.
	Classical continuum models incorporating nonlinear rheology \eqref{powerlawintro} into Navier--Stokes-type equations suffer from a well-known parabolic paradox: they predict an infinite propagation speed for disturbances, a physically untenable artifact when dealing with fast transients and shock formation.
	A mathematically rigorous and thermodynamically consistent resolution is provided by \emph{Rational Extended Thermodynamics} (RET) \cite{MullerRuggeri1998, RuggeriSugiyama2021}.

	RET was originally developed from the hierarchy of moments in kinetic theory.
	In that setting, the flux associated with a retained moment is itself a
	higher-order moment and, once admitted among the independent variables, is
	supplied with the next balance equation of the hierarchy. A general continuum
	theory need not possess such a literal moment sequence, but RET preserves its
	structural idea: a nonequilibrium flux added to the state space is treated as
	an independent field governed by an additional balance law, whose density,
	flux, and production are local objective constitutive functions, rather than
	being prescribed through a nonlocal constitutive equation. The enlarged
	system is then constrained by the entropy principle and thermodynamic
	stability. In particular, entropy production must be nonnegative, while the
	physical entropy must be concave---equivalently, the mathematical entropy must
	be convex---with respect to the balance densities. The corresponding main
	field therefore provides symmetrizing variables, so that the homogeneous part
	of the system is symmetric hyperbolic
	\cite{MullerRuggeri1998,RuggeriSugiyama2021,RuggeriStrumia1981,Ruggeri2024}.

	Recently, Ruggeri \cite{Ruggeri2025} proposed a \emph{hyperbolic relaxation variant}
	of the classical power-law theory for non-Newtonian fluids. In this RET model the
	viscous stress tensor is promoted to an independent field with a finite relaxation
	time; only in the fast-relaxation limit does one recover the usual parabolic
	power-law constitutive response. The power-law part is parameterized by a flow
	index $m > 0$ and a consistency coefficient $k > 0$. The model reduces to the
	classical Newtonian viscous fluid for $m = 1$, and describes shear-thinning
	($m < 1$) or shear-thickening ($m > 1$) behavior otherwise.
	The present model belongs to a series of studies by Ruggeri and coworkers in which
	the universal principles of RET have been applied to models of nonlinear
	viscoelasticity \cite{Ruggeri2024,ArimaRuggeriTaniguchi2026,Amabili2025,Giusti}.
	Related wave-propagation issues, including acceleration waves and the K-condition, were examined in \cite{RuggeriAccel}.
	The isothermal restriction was eliminated in the recent paper of Arima and Ruggeri \cite{ArimaRuggeri2026}. The resulting nonisothermal equations retain the \emph{symmetric hyperbolic} structure required by RET, ensuring finite propagation speeds and allowing the standard local well-posedness theory for the Cauchy problem \cite{FischerMarsden1972,Kato1975,Majda1984}.
	In a mathematically related framework, but based on a different physical interpretation, symmetric hyperbolic formulations have been proposed by various authors to model elastic, viscous, and relaxation effects; in particular, we mention~\cite{GodunovRomenski1972,PeshkovRomenski2016,DumbserPeshkovRomenski2017}.

\subsection{Aim and main conclusions}
The shock-wave structure of the isothermal model was analyzed in the
companion paper with Taniguchi~\cite{RuggeriTaniguchi}, now published in
\emph{Physics of Fluids}. That study identifies two thresholds for the
flow-behavior index: a weak-shock threshold at $m=2$ and a near-subshock
threshold at $m=1$. In the isothermal theory the temperature is fixed and
the energy balance is omitted; consequently, the pressure law and
consistency coefficient are evaluated at a prescribed temperature and
there is no thermomechanical feedback along the traveling wave.

The passage to the present nonisothermal problem is not obtained by simply
replacing constants in the isothermal profile equation by
temperature-dependent coefficients. The energy equation imposes the
algebraic constraint $\theta=\Theta(F;s)$, changes the pressure and stress
into composite functions along the Hugoniot branch, modifies the relaxation
production, and enters the traveling-wave denominator through the
nonequilibrium characteristic speed. In fact,
Corollary~\ref{cor:no-constant-temperature} proves that, under the present
thermodynamic and convexity assumptions, no nontrivial compressive profile
of the full system can remain at constant temperature. Thus an isothermal
shock profile is not a hidden special solution of the nonisothermal
traveling-wave problem.

The main mathematical novelty is the global characteristic-ordering
identity proved in Lemma~\ref{lem:Dpositive}. For equations of state with
$p_\theta\ge0$ and strictly convex reduced Hugoniot pressure
$P_s(F)=p(F,\Theta(F;s))$, it shows that the positive nonequilibrium
characteristic speed is strictly larger at every interior point of the
compressive profile than at the upstream state. Hence the first
characteristic encounter occurs precisely at $F=1$ when $M_0=\Mst$, and a
continuous monotone profile exists throughout $1<M_0<\Mst$ without imposing
a separate regular-branch hypothesis. For $M_0>\Mst$, the
Boillat--Ruggeri theorem excludes a $C^1$ structure; just above threshold,
the reduced equation also exhibits the corresponding uncancelled
characteristic pole.

With this global issue resolved, the asymptotic analysis proves that the two
isothermal indices survive the full thermomechanical coupling. The
weak-shock exponent remains $1-2/m$, so $m=2$ separates broadening, a finite
plateau, and collapse. Near $\Mst$, the local degeneracy of the coupled
traveling-wave denominator produces the threshold $m=1$ and, for $m>1$,
the exponent $(m-1)/m$. These exponents are constitutive-independent within
the present class, whereas the finite limiting thicknesses and prefactors
retain the full dependence on the equation of state, internal energy,
temperature profile, and temperature-dependent consistency coefficient.
The Tait--Murnaghan example additionally displays a genuinely thermal
effect: when the dimensional viscous--relaxation scale is fixed, increasing
the reference temperature changes the upstream adiabatic sound speed,
lowers the critical Mach number, and, for the thermally thinning law used
here, reduces the resolved shock thickness.

The manuscript is mathematically self-contained; the published companion
paper is used to identify the isothermal benchmark, not to assume convergence
of one family of profiles to the other. Throughout this paper, the word
``universal'' refers only to the two thresholds and their associated critical
exponents within the present class of symmetric-hyperbolic power-law
relaxation closures. Finite limiting thicknesses and scaling amplitudes are
not universal numerical constants.

\subsection{Structure of the paper}
The paper is organized as follows. In Section~\ref{sec:model} we recall
	the nonisothermal RET model of Arima and Ruggeri~\cite{ArimaRuggeri2026}
	and reduce the traveling-wave equations to a scalar ODE by exploiting
	the conservation laws. The constitutive hypotheses (H1)--(H4b) are
	stated precisely in Section~\ref{sec:hyp}. Section~\ref{sec:main}
	contains the main analytical results. We first identify the critical
	Mach number \(\Mst\) and prove, in Theorem~\ref{thm:existence},
	the existence of a monotone continuous traveling-wave profile for
	\(1<M_0<\Mst\); the local subshock state beyond this threshold is
	described in Proposition~\ref{prop:subshock}. We then prove the two
	universality theorems: Theorem~\ref{thm:weak} establishes the
	weak-shock scaling and Theorem~\ref{t:nearcrit} the near-critical
	behavior. Section~\ref{sec:constitutive} discusses
	the physical constitutive functions for non-Newtonian fluids and
	verifies the hypotheses. Section~\ref{sec:numerical} collects the
	numerical setup and illustrates the theoretical results, including
	the role of the reference temperature $\theta_0$ on $\Mst$ and on
	the dimensionless shock thickness $\Delta$. Concluding remarks are given in
	Section~\ref{sec:conclusions}.
	
	\section{The nonisothermal RET model and traveling-wave reduction}
	\label{sec:model}

	\subsection{Nonisothermal RET system}
We consider a one-dimensional non-Newtonian power-law fluid in Lagrangian coordinates $(X,t)$.
If $x=x(X,t)$ denotes the motion, the one-dimensional deformation gradient is
$F=\partial x/\partial X>0$. In multidimensional notation the Jacobian satisfies
$J=\det\mathbf F=\rho_0/\rho$; in the present one-dimensional reduction this gives
$F=J=\rho_0/\rho$, with $\rho_0$ the constant reference mass density. The unknown
fields are the velocity $v(X,t)$, the deformation gradient $F(X,t)>0$, the absolute
temperature $\theta(X,t)>0$, and the scalar nonequilibrium stress $\sigma(X,t)$,
while $p(F,\theta)$ and $\eps(F,\theta)$ denote the pressure and the specific
internal energy. The variable $\sigma$ is the one-dimensional viscous Piola stress
already introduced in the Introduction; in one space dimension the Piola and Cauchy
viscous stresses coincide, since the cofactor $F^C$ is equal to $1$ by convention.
With this notation, the nonisothermal RET system reads~\cite{ArimaRuggeri2026}:
	\begin{subequations}\label{system}
		\begin{align}
			&\rhoo\,v_t + \bigl[p(F,\theta)-\sigma\bigr]_X = 0,\label{system:a}\\
			&F_t - v_X = 0,\label{system:b}\\
			&\rhoo\!\left(\dfrac{v^2}{2}+\eps(F,\theta)\right)_t
			+\bigl[(p(F,\theta)-\sigma)\,v\bigr]_X = 0,\label{system:c}\\
			&\bigl[Z(\bar\sigma)\bigr]_t - v_X
			= P(F,\sigma,\theta),\label{system:d}
		\end{align}
	\end{subequations}
	where
	\begin{equation*}
		P(F,\sigma,\theta)
		= -F\,2^{1/m-1}\,k(\theta)^{-1/m}\,|\sigma|^{(1-m)/m}\,\sigma,
	\end{equation*}
	and $\bar\sigma=\sigma/\theta$. The function $Z$ is chosen with $Z(0)=0$ and
	\begin{equation}\label{Zomega}
		\frac{dZ}{d\bar\sigma}(\bar\sigma)=\omega(\bar\sigma)>0,
	\end{equation}
	so that \eqref{system:d} is equivalently
\begin{equation}\label{2dd}
\omega(\bar\sigma)\bar\sigma_t-v_X=P(F,\sigma,\theta).
    \end{equation}
 The coefficient
	$\omega(\bar\sigma)>0$ is not an arbitrary material multiplier: it is related to the
	nonequilibrium part of the entropy. More precisely, we denote by
	$\eta(F,\theta,\bar\sigma)$ the full extended specific entropy and by
	$\eta^{\rm eq}(F,\theta)$ its equilibrium restriction. They are related by
	\begin{equation}\label{extendedentropy}
		\eta(F,\theta,\bar\sigma)
		=\eta^{\rm eq}(F,\theta)
		-\frac{1}{\rhoo}\int_0^{\bar\sigma}\omega(z)z\,dz,
	\end{equation}
	where the equilibrium entropy satisfies the Gibbs relation
	\begin{equation}\label{gibbsrelation}
		\theta\,d\eta^{\rm eq}=d\eps+\frac{p}{\rhoo}\,dF.
	\end{equation}
	Equivalently, the equilibrium part of the entropy principle may be written
	in terms of an equilibrium Helmholtz free energy
	\begin{equation}\label{freeenergydef}
		\psi(F,\theta)=\eps(F,\theta)-\theta\eta^{\rm eq}(F,\theta).
	\end{equation}
	For the fluid specialization considered here, it is convenient to write the
	thermodynamic closure directly in terms of the pressure. The equilibrium
	Helmholtz free energy satisfies
	\begin{equation}\label{freeenergyclosure}
		\eps=\psi-\theta\psi_\theta,
		\qquad
		p=-\rhoo\psi_F.
	\end{equation}
	Differentiating these relations gives the useful identities
	\begin{equation}\label{freeenergyidentities}
		\eps_F=\frac{\theta p_\theta-p}{\rhoo},
		\qquad
		\eps_\theta=-\theta\psi_{\theta\theta}.
	\end{equation}
	Thus the thermodynamic convexity conditions inherited from the entropy
	principle become, in the pressure notation,
	\begin{equation}\label{thermoconvpressure}
		p_F<0,
		\qquad
		\eps_\theta>0.
	\end{equation}
	Then the balance laws~\eqref{system:a}--\eqref{2dd} imply the
	entropy law
	\begin{equation}\label{entropylaw}
		\rhoo \eta_t
		= -\frac{\sigma}{\theta}\,P(F,\sigma,\theta)\ge0.
	\end{equation}
	Here $\eta$ denotes the physical specific entropy, so the sign in
	\eqref{entropylaw} is the physical entropy-production sign; the convex
	mathematical entropy used to symmetrize the system is $-\eta$.  The entropy
	flux is absent because no heat flux is included. For the
	power-law production term above one has explicitly
	\begin{equation*}
		-\sigma P(F,\sigma,\theta)
		=F\,2^{1/m-1}\,k(\theta)^{-1/m}|\sigma|^{(m+1)/m}\ge0,
	\end{equation*}
	so that the second law is satisfied. Thus, for the chosen fields and in
	the absence of heat flux, the relaxation equation~\eqref{2dd} is the
	RET balance law compatible with the entropy principle; in particular,
	$\bar\sigma=\sigma/\theta$ is the natural thermodynamic nonequilibrium
	variable.

	Under the standard thermodynamical convexity conditions the physical entropy is
	concave; equivalently, the mathematical entropy $-\eta$ is convex. With the main field 
    \begin{equation*}
		\mathbf{u}' \equiv \frac{1}{\theta}(v,-p,-1,\sigma),
	\end{equation*}
    the system becomes symmetric hyperbolic according to the general theory \cite{RuggeriStrumia1981}.

	In the fast-relaxation limit $\omega\to 0$, equation~\eqref{2dd} reduces to
	the inverse of the one-dimensional power-law relation associated with the
	Cauchy viscous stress in~\eqref{powerlawintro}. In the present paper we
	consider the particular case of a quadratic nonequilibrium entropy, for
	which $\omega$ is constant \cite{ArimaRuggeri2026}:
	\begin{equation}\label{omegaConstant}
		\omega = \frac{\tau_0 \theta_0}{\mu_0},
	\end{equation}
	and, with the normalization $Z(0)=0$, equivalently
	$Z(\bar\sigma)=\omega\bar\sigma$. Here $\tau_0, \theta_0, \mu_0$
	are positive constants having dimensions
	respectively of a time, a reference temperature, and a viscosity. Since
	$[\mu_0]=[\text{stress}][\text{time}]$ and
	$[\bar\sigma]=[\text{stress}]/[\text{temperature}]$, one has
	$[\omega]=[\text{temperature}]/[\text{stress}]$, as required by
	$Z'(\bar\sigma)=\omega$.

	In the present model, the macroscopic heat flux is deliberately neglected.
	The purpose is to isolate the thermo-viscous coupling generated by the
	energy balance and viscous dissipation within the minimal four-field RET
	system. On the short mechanical time scales associated with shock
	propagation, thermal diffusion may be subdominant to this local heating.
	This approximation can be quantified dimensionally. If
	$\chi=\kappa/(\rho_0 C_v)$ is the thermal diffusivity and
	$\ell_0=\tau_0c_0$ is the natural mechanical relaxation length introduced
	below, the ratio of Fourier diffusion to transport across a layer of width
	$\ell_0$ is
	\begin{equation}\label{thermalPeclet}
		\varepsilon_q:=\frac{\chi}{c_0\ell_0}
		=\frac{\chi}{\tau_0c_0^2}=\mathrm{Pe}_0^{-1}.
	\end{equation}
	Thus neglect of Fourier conduction is asymptotically consistent in the
	mechanically dominated regime $\mathrm{Pe}_0\gg1$, equivalently when the
	advective--diffusive length $\chi/c_0$ is small compared with $\ell_0$
	(or when $\sqrt{\chi\tau_0}\ll\ell_0$).  Because $\kappa$ is not a
	coefficient of the reduced four-field model, the normalized computations
	below should be interpreted within this regime rather than as a claim of
	validity for $\mathrm{Pe}_0=O(1)$. No single numerical value of
	$\mathrm{Pe}_0$ is assigned here, because a material-specific estimate
	would require both the thermal diffusivity and the RET relaxation time
	$\tau_0$ measured for the same material and upstream state. Thus
	\eqref{thermalPeclet} is an applicability criterion to be checked when the
	model is calibrated to a particular fluid.
	Moreover, adding a classical Fourier law would destroy finite propagation
	speed and hence the symmetric-hyperbolic character of the model. A fully
	hyperbolic treatment of heat transport is possible in RET by introducing
	the heat flux as an additional independent nonequilibrium field, but this
	would add further characteristic families. Whether those additional
	characteristics preserve the two critical thresholds found here is a
	separate problem and is not assumed in the present analysis.

	\subsection{Traveling-wave reduction}
	We seek right-going traveling waves depending on the single variable
	\begin{equation}\label{travelingcoordinate}
		\varphi=X-st,
		s>0.
	\end{equation}
	The equilibrium constant boundary data are imposed as
	\begin{equation}\label{TWboundary}
		(F,v,\theta,\sigma)(\varphi)\to
	\begin{cases}
		(F_-,v_-,\theta_-,0), & \varphi\to-\infty,\\[0.25em]
		(1,0,\theta_0,0), & \varphi\to+\infty.
	\end{cases}
	\end{equation}
	The state at $\varphi\to+\infty$ is the unperturbed upstream state. The case
	$s<0$ is obtained by reversing the propagation direction and is not considered
	separately. Integration of the conservation laws \eqref{system:a}--\eqref{system:c},
	normalized by the upstream state in \eqref{TWboundary}, yields
	\begin{align}
		& v = s(1-F), \label{mass}\\
		&\sigma = \Si(F,\theta)
		\equiv p(F,\theta)-p_0-\rhoo s^2(1-F),
		\label{momentum}\\
		&\calE(F,\theta) = 0, \label{energy}
	\end{align}
	where $p_0=p(1,\tho)$ and $\calE(F,\theta)=0$ is the integrated energy
	equation (an implicit relation between $F$ and $\theta$ derived below).
	Evaluating \eqref{mass}--\eqref{energy} at the downstream equilibrium
	state, where $\sigma_-=0$, gives the Rankine--Hugoniot relations for the
	equilibrium subsystem obtained by setting $\sigma=0$ in
	\eqref{system:a}--\eqref{system:c}:
	\begin{align}
		&v_- = s(1-F_-), \label{RHmass}\\
		&p(F_-,\theta_-)-p_0=\rhoo s^2(1-F_-), \label{RHmom}\\
		&\eps(F_-,\theta_-)-\eps_0
		=\frac{p_0}{\rhoo}(1-F_-)+\frac{s^2}{2}(1-F_-)^2.\label{RHenergy}
	\end{align}
	It is useful to parametrize the shock branch by the downstream
	deformation $F_-$, rather than by the shock speed $s$.  
    
    Following Arima and Ruggeri~\cite{ArimaRuggeri2026}, the equilibrium
	adiabatic sound speed is defined by
	\begin{equation}\label{adiabaticSound}
		c_S^2(F,\theta)
		:=\left.\frac{\partial p}{\partial\rho}\right|_{\eta^{\rm eq}}
		=\frac{F^2}{\rhoo}
		\left(
		-p_F(F,\theta)
		+\frac{\theta p_\theta^2(F,\theta)}{\rhoo\eps_\theta(F,\theta)}
		\right).
	\end{equation}
	The positive Lagrangian characteristic velocity of the equilibrium
	subsystem is therefore
	\begin{equation}\label{equilibriumlambda}
		\bar\lambda(F,\theta)=\frac{c_S(F,\theta)}{F}.
	\end{equation}
	At the upstream equilibrium state, where
	$\bar\lambda(1,\theta_0)=c_S(\theta_0)$, we define once and for all
	\begin{equation}\label{MachDef}
		c_S(\theta_0):=c_S(1,\theta_0),
		\qquad
		M_0:=\frac{s}{c_S(\theta_0)}.
	\end{equation}
	The Rankine--Hugoniot relations above are precisely the classical
	Euler shock relations written in Lagrangian variables. In what follows
	we do not regard the orientation of the equilibrium shock as a new result:
	we select the standard entropy-admissible right-going Lax branch of a
	genuinely nonlinear acoustic family; see, for instance,
	Refs.~\cite{CourantFriedrichs1948,Lax1957,MenikoffPlohr1989}. This is the
	ordinary compressive branch. Since $F=\rho_0/\rho$ and the upstream state
	is normalized by $F=1$, compression gives
	\begin{equation}\label{LaxCompressionF}
		F_-<1.
	\end{equation}
	The same classical admissible branch has increased pressure, temperature,
		and entropy across the shock. We use this standard ordering only as a
		branch selection, not as a new result. Thus, for the present normalization,
	\begin{equation}\label{DownstreamOrdering}
		p(F_-,\theta_-)>p_0,
		\qquad
		\theta_->\theta_0,
		\qquad
		v_-=s(1-F_-)>0,
	\end{equation}
	where the last inequality follows directly from \eqref{RHmass} since
	$s>0$ and $F_-<1$. The corresponding Lax inequalities for the positive
	equilibrium family are
	\begin{equation}\label{LaxEquilibriumInequality}
		\bar\lambda(1,\theta_0)<s<\bar\lambda(F_-,\theta_-).
	\end{equation}
	By \eqref{MachDef}, the upstream part of
	\eqref{LaxEquilibriumInequality} is simply $M_0>1$.
	Eliminating $s^2$ between \eqref{RHmom} and \eqref{RHenergy} gives the
	scalar thermodynamic Hugoniot relation
	\begin{equation}\label{RHthermal}
		\eps(F_-,\theta_-)-\eps_0
		=\frac{1-F_-}{2\rhoo}\,
		\bigl(p(F_-,\theta_-)+p_0\bigr).
	\end{equation}
	Anticipating the temperature compatibility assumption (H2) below,
	\eqref{RHthermal} determines the downstream temperature $\theta_-$
	implicitly as a function of the chosen value of $F_-$. Once $\theta_-$
	is fixed, \eqref{RHmom} gives
	\begin{equation}\label{shockSpeedFromFminus}
		s^2=
		\frac{p(F_-,\theta_-)-p_0}{\rhoo(1-F_-)}.
	\end{equation}
	The positivity of \eqref{shockSpeedFromFminus} is therefore consistent
	with the ordering \eqref{DownstreamOrdering}. Thus the intervals naturally
	associated with the traveling wave are $[F_-,1]$ for the deformation and
	$[\theta_0,\theta_-]$ for the temperature.

	\paragraph{Energy integral.}
	Using \eqref{mass} and the momentum integral to eliminate $v$ and
	the flux $(p(F,\theta)-\sigma)v$, the energy conservation law integrates exactly to a purely algebraic relation for the local internal energy $\eps$ along the shock profile:
	\begin{equation}\label{energyint}
		\eps(F, \theta) - \eps_0 = \frac{p_0}{\rhoo}(1 - F) + \frac{s^2}{2}(1 - F)^2.
	\end{equation}
	
	Under Assumption~\ref{ass:state} below,
	equation~\eqref{energyint} defines $\theta$ as an explicit function
	$\theta=\Th(F;s)$ of $F$ alone along the profile.
	Substituting $\theta=\Th(F;s)$ into \eqref{momentum}, we define the
	reduced composite stress and its temperature-normalized counterpart by
	\begin{equation}\label{ReducedSigmaPi}
		\Si(F;s):=\Si\bigl(F,\Th(F;s)\bigr),\qquad
		\Pi(F;s):=\frac{\Si(F;s)}{\Th(F;s)}.
	\end{equation}
	From now on, a prime on a reduced function of $F$, such as $\Theta'(F;s)$,
	$\Pi'(F;s)$, or $\Si'(F;s)$, denotes differentiation with respect to $F$.
	The derivative with respect to the traveling-wave variable $\varphi$ will be
	written as $F_\varphi=dF/d\varphi$.
	The relaxation equation \eqref{2dd} then becomes the scalar autonomous ODE
	\begin{equation}\label{ODE}
		\boxed{F_\varphi \equiv \frac{dF}{d\varphi}= \frac{\calN(F)}{sD(F)}},
		\qquad
		D(F)=1-\frac{\tau_0\theta_0}{\mu_0}\,\Pi'(F),
	\end{equation}
	Since $F$ is dimensionless and $\Pi=\Sigma/\Theta$ has dimensions of
	stress divided by temperature, the product
	$(\tau_0\theta_0/\mu_0)\Pi'(F)$ is dimensionless; hence $D(F)$ is
	dimensionless, consistently with the characteristic identity
	\eqref{Dcharacteristic} below. The numerator is
	\begin{equation}\label{numerator}
		\calN(F) = -F\,2^{1/m-1}\,k\!\bigl(\Th(F)\bigr)^{-1/m}
		\,|\Si(F)|^{(1-m)/m}\,\Si(F).
	\end{equation}
	The scalar traveling-wave problem is completed by the boundary conditions
	\begin{equation}\label{BCscalar}
		F(\varphi)\to F_- \quad (\varphi\to-\infty),
		\qquad
		F(\varphi)\to 1 \quad (\varphi\to+\infty).
	\end{equation}
	The denominator in \eqref{ODE} has a useful characteristic interpretation.
	Differentiating the energy integral \eqref{energyint} and \eqref{freeenergyidentities} gives along the Hugoniot branch
\begin{equation}\label{ThetaPrimeGeneral}	\Theta'(F;s)
		=
		\frac{\Si(F;s)-\Theta(F;s)p_\theta(F,\Theta(F;s))}
		{\rhoo\eps_\theta(F,\Theta(F;s))}.
	\end{equation}
	Consequently,
	\begin{equation}\label{PiPrimeIdentity}
		\Theta\,\Pi'
		=
		p_F+\rhoo s^2
		-
		\frac{\bigl(\Si-\Theta p_\theta\bigr)^2}
		{\rhoo\Theta\eps_\theta},
	\end{equation}
	where all constitutive functions are evaluated at
	\((F,\Theta(F;s))\). Comparing \eqref{PiPrimeIdentity} with the square of the
	positive acoustic characteristic speed of the RET system \cite{ArimaRuggeri2026}
	\begin{equation}\label{lambda0local}
		\lambda_0^2(F;s)
		=
		-\frac{p_F}{\rhoo}
		+
		\frac{\bigl(\Si-\Theta p_\theta\bigr)^2}
		{\rhoo^2\Theta\eps_\theta}
		+
		\frac{\mu_0\Theta}{\rhoo\tau_0\theta_0},
	\end{equation}
	we obtain the exact identity
	\begin{equation}\label{Dcharacteristic}
		D(F;s)
		=1-\frac{\tau_0\theta_0}{\mu_0}\Pi'(F;s)
		=
		\frac{\rhoo\tau_0\theta_0}{\mu_0\Theta(F;s)}
		\left(\lambda_0^2(F;s)-s^2\right).
	\end{equation}
	Thus the denominator can vanish only when the shock speed meets the
	positive acoustic eigenvalue. At this stage this is only a local
	characteristic identity for the reduced traveling-wave equation; its role
	in the critical Mach number and in subshock formation is discussed below,
	where it is actually used.

	\section{Constitutive hypotheses}
	\label{sec:hyp}

	\begin{assumption}[Admissible thermodynamics and temperature range]\label{ass:state}
		The constitutive functions $p(F,\theta)$ and $\eps(F,\theta)$ are
		$C^2$ on a neighborhood of $[F_-,1]\times[\tho,\theta_-]$ and satisfy:
		\begin{enumerate}
			\item[(H1)] \emph{Thermodynamic consistency and stability}: the
				equilibrium Gibbs relation~\eqref{gibbsrelation} holds for
				$\eta^{\rm eq}$; equivalently, there exists an equilibrium Helmholtz free
				energy $\psi(F,\theta)$ satisfying~\eqref{freeenergyclosure}. We work in a
				thermodynamically stable region, in the standard sense of concavity of the
				physical equilibrium entropy (or, equivalently, convexity of the mathematical
				entropy obtained by changing its sign). In the pressure representation this
				requires, on the relevant states,
\begin{equation}\label{thermostabilityH1}
p_F(F,\theta)<0,
\qquad
\eps_\theta(F,\theta)=C_v(F,\theta)>0.
\end{equation}
The adiabatic sound speed $c_S$ defined in \eqref{adiabaticSound} is therefore
positive; we assume it is finite on the relevant compact range of states. We
additionally restrict attention to an ordinary thermal-expansion region,
\begin{equation}\label{ordinaryThermalExpansion}
p_\theta(F,\theta)\ge0.
\end{equation}
This sign condition is not a consequence of thermodynamic stability. At fixed
$F$ the density is fixed, so $p_\theta$ is the usual isochoric thermal-pressure
coefficient. If
\[
K_T:=\rho\left(\frac{\partial p}{\partial\rho}\right)_\theta>0,
\qquad
\alpha:=-\frac{1}{\rho}
\left(\frac{\partial\rho}{\partial\theta}\right)_p,
\]
denote the isothermal bulk modulus and the volumetric thermal-expansion
coefficient, respectively, then the exact identity
\[
\left(\frac{\partial p}{\partial\theta}\right)_\rho=\alpha K_T
\]
shows that \eqref{ordinaryThermalExpansion} is precisely the ordinary
thermal-expansion condition $\alpha\ge0$. It excludes anomalous intervals,
such as water near its density maximum, and negative-thermal-expansion
materials. The upstream sound speed and Mach number are those defined in
\eqref{MachDef}.
				\item[(H2)] \emph{Temperature compatibility and branch regularity}: for
			each selected right-going compressive Euler branch, the scalar Hugoniot relation
			\eqref{RHthermal} determines a locally unique admissible downstream
			temperature $\theta_-$. The ordering $\theta_->\theta_0$ is the one fixed by
			the selected classical branch in \eqref{DownstreamOrdering}. Moreover, on a
			neighborhood of each compact shock family considered below, the algebraic
			energy integral~\eqref{energyint} admits a solution
			$\theta=\Th(F;s)$ that is jointly $C^2$ in $(F,s)$. On each individual
			branch, $\Th(1;s)=\theta_0$, $\Th(F_-;s)=\theta_-$, and
			$\Th(F;s)\in[\theta_0,\theta_-]$ for all $F\in[F_-,1]$.
			
	\end{enumerate}
\end{assumption}

\begin{assumption}[Power-law rheology]\label{ass:k}
\begin{enumerate}
\item[(H3)] The flow index \(m>0\) is fixed. The consistency coefficient
\(k(\theta)\) is assumed to be continuous and strictly positive on the
admissible temperature interval. More precisely, there exist constants
\(0<k_*\le k^*<\infty\) such that
\[
        k_*\le k(\theta)\le k^*,
        \qquad \theta\in[\theta_0,\theta_-].
\]
\end{enumerate}
\end{assumption}

\begin{assumption}[Equilibrium branch and Hugoniot convexity]
\label{ass:hugoniot-convexity}
\begin{enumerate}
\item[(H4a)] \emph{Classical genuinely nonlinear Euler branch.}
The nonisothermal equilibrium $p$-system, obtained from
\eqref{system:a}--\eqref{system:c} by setting $\sigma=0$, is assumed
to be strictly hyperbolic on the stable thermodynamic region specified
in \textup{(H1)}, and its positive acoustic field is genuinely nonlinear
with the classical fluid orientation. In fluid-dynamical terms, this means that the fundamental derivative has the standard
sign on the selected branch, so that nonclassical Bethe--Zel'dovich--Thompson effects, such as
rarefaction shocks, are excluded. Throughout the paper we select the
standard entropy-admissible right-going Lax compressive branch fixed by
\eqref{LaxCompressionF}, \eqref{DownstreamOrdering}, and
\eqref{LaxEquilibriumInequality}.

\item[(H4b)] \emph{Hugoniot convexity.}
For every nontrivial right-going compressive equilibrium Hugoniot branch
issuing from $(F,\theta)=(1,\theta_0)$, the reduced pressure
\[
P_s(F):=p\bigl(F,\Theta(F;s)\bigr)
\]
is strictly convex on $[F_-,1]$:
\[
P_s''(F)>0,\qquad F\in[F_-,1].
\]
Equivalently, because the Rankine--Hugoniot chord
$p_0+\rhoo s^2(1-F)$ is affine in $F$, the reduced composite stress
$\Si(F;s)$ defined in \eqref{ReducedSigmaPi} satisfies
\[
\Si''(F;s)>0,\qquad F\in[F_-,1].
\]
For the corresponding isothermal pressure law this becomes the standard condition
$p_{FF}(F,\theta_0)>0$.
\end{enumerate}
\end{assumption}

Introduce the dimensionless viscous--relaxation parameter
\begin{equation}\label{bmu0def}
\bmuo:=\frac{\mu_0}{\tau_0\rhoo c_S^2(\tho)}.
\end{equation}
The sign condition \eqref{ordinaryThermalExpansion} and Hugoniot
convexity will imply, through Lemma~\ref{lem:Dpositive}, both the absence
of an earlier interior characteristic encounter and the transversality
of the terminal upstream encounter. Thus no separate regular-branch
hypothesis is required for this constitutive class. This conclusion is
specific to the structural assumptions stated here: entropy-compatible
hyperbolic relaxation systems outside this class can form subshocks below
the maximum unperturbed characteristic speed~\cite{TaniguchiRuggeri2018}.

We record one immediate consequence of the Rankine--Hugoniot relations,
	which also fixes the meaning of the upstream Mach number used below. With
	the notation $P_s$ introduced in \textup{(H4b)}, \eqref{RHmom} gives
	\begin{equation}\label{secantspeed}
		s^2
		=
		\frac{P_s(F_-)-P_s(1)}{\rhoo(1-F_-)}
		=
		-\frac{1}{\rhoo}\,
		\frac{P_s(F_-)-P_s(1)}{F_--1}.
	\end{equation}
	Moreover, differentiating \eqref{energyint} at the upstream state gives
	\begin{equation}\label{PsprimeUpstream}
		P_s'(1)=-\rhoo c_S^2(\tho).
	\end{equation}
	Since $F_-<1$ on the selected right-going Lax compressive branch
	\eqref{LaxCompressionF}, and since $P_s$ is strictly convex by
	Assumption~\ref{ass:hugoniot-convexity}\,\textup{(H4b)}, the secant slope in
	\eqref{secantspeed} is strictly smaller than $P_s'(1)$. Therefore
	\begin{equation}\label{LaxMgreaterOne}
		s^2>c_S^2(\tho),
		\qquad\text{hence}\qquad M_0>1.
	\end{equation}
	Thus the restriction $M_0>1$ is not an additional convention: it is precisely
	the upstream part of the right-going Lax condition \eqref{LaxEquilibriumInequality},
	written in Mach-number form.

	The properties of the reduced composite stress $\Si(F;s)$ defined in
	\eqref{ReducedSigmaPi}, needed for the thickness analysis, are consequences of
	Assumptions~\ref{ass:state}--\ref{ass:hugoniot-convexity}
	and the Rankine--Hugoniot conditions.
	We state and prove them now.
	
	\begin{theorem}[Properties of $\Si$]\label{thm:Sigma}
		Under Assumptions~\ref{ass:state} and~\ref{ass:hugoniot-convexity},
		consider a selected compressive shock branch with $s>c_S(\tho)$
		(i.e.\ $M_0>1$). The reduced composite stress
		$\Si(F;s)$ defined in \eqref{ReducedSigmaPi} satisfies:
		\begin{enumerate}
			\item[\textup{(P1)}] $\Si(1)=0$ and $\Si(F_-)=0$, where $F_-$ is the
			downstream state determined by the Rankine--Hugoniot conditions.
			\item[\textup{(P2)}] $\Si(F)<0$ for all $F\in(F_-,1)$.
			\item[\textup{(P3)}] $\Si'(1) = \rhoo \bigl(s^2 - c_S^2(\tho)\bigr) > 0$.
			\item[\textup{(P4)}] $\Si'(F_-)<0$.
			\item[\textup{(P5)}] $\Si''(F)$ is bounded on $[F_-,1]$.
		\end{enumerate}
	\end{theorem}
	
	\begin{proof}
		\textbf{(P1).}
		By \eqref{momentum} and \eqref{ReducedSigmaPi}, at the upstream state
		$F=1$, $\theta=\tho$, one has $\Si(1)=p_0-p_0=0$.
		The downstream state is the selected compressive state $F_-<1$ satisfying
		the Rankine--Hugoniot relations \eqref{RHmass}--\eqref{RHenergy}; in
		particular \eqref{RHmom}, with $\theta_-=\Th(F_-;s)$, gives
		$\Si(F_-)=0$.
		
		\textbf{(P2).}
		By Assumption~\ref{ass:hugoniot-convexity}\,\textup{(H4b)}, $\Si''(F)>0$ on $[F_-,1]$,
		so $\Si$ is strictly convex on that interval.
		Since $\Si(F_-)=0$ and $\Si(1)=0$ (from (P1)), strict convexity gives:
		for any $F = \alpha F_- + (1-\alpha)\cdot 1$ with $\alpha\in(0,1)$,
		\[
			\Si(F)<\alpha\,\Si(F_-)+(1-\alpha)\,\Si(1)=0.
		\]
		Hence $\Si(F)<0$ for all $F\in(F_-,1)$.\medskip
		
		\textbf{(P3).}
		We compute $\Si'(1)$ directly by differentiating
		\eqref{ReducedSigmaPi}:
		\begin{equation}\label{Siprime1}
			\Si'(1) = p_F(1,\tho) + p_\theta(1,\tho)\,\Th'(1) + \rhoo s^2.
		\end{equation}
		
		\textit{Step 1: compute $\Th'(1)$.}
		Differentiating the energy integral~\eqref{energyint} with respect to $F$,
		and using the thermodynamic identity
		$\eps_F=(\theta p_\theta-p)/\rhoo$, we obtain
		\begin{equation}\label{Thetaprime}
			\frac{\theta p_\theta(F,\theta)-p(F,\theta)}{\rhoo}
			+C_v(F,\theta)\Th'(F)
			= -\frac{p_0}{\rhoo}-s^2(1-F).
		\end{equation}
		Evaluating this identity at $F=1$, where $\theta=\tho$, gives
		\begin{equation*}
			\Th'(1)
			= -\frac{\tho p_\theta(1,\tho)}{\rhoo C_v(1,\tho)}.
		\end{equation*}
		
		\textit{Step 2: use the adiabatic sound speed.}
		Substituting this expression into~\eqref{Siprime1}, we find
		\begin{align*}
			\Si'(1)
			&=p_F(1,\tho)
			-\frac{\tho p_\theta^2(1,\tho)}{\rhoo C_v(1,\tho)}
			+\rhoo s^2  \\
			&=\rhoo\bigl(s^2-c_S^2(\tho)\bigr),
		\end{align*}
		where the last equality follows from the definition of the upstream
		adiabatic sound speed in~\eqref{adiabaticSound}. Thus
		$\Si'(1)>0$ is exactly equivalent to the physical entropy condition
		$s>c_S(\tho)$, or equivalently $M_0>1$ by \eqref{MachDef}.
		
		\textbf{(P4).}
		Because $\Si''>0$, the derivative $\Si'$ is strictly increasing. Since
		$\Si(1)-\Si(F_-)=\int_{F_-}^{1}\Si'(F)\,dF=0$, the derivative must
		change sign on the interval. Together with \textup{(P3)}, this yields
		$\Si'(F_-)<0$.

		\textbf{(P5).}
		Since $p$ is $C^2$ and $\Th$ is $C^2$ on $[F_-,1]$ by
		\textup{(H1)}--\textup{(H2)}, the reduced stress in
		\eqref{ReducedSigmaPi} is $C^2$ on $[F_-,1]$; hence $\Si''$ is
		continuous and bounded on this compact interval.
	\end{proof}
	
	\begin{remark}
		Property (P1) records the Rankine--Hugoniot endpoint conditions;
		properties (P3) and (P5) follow from the thermodynamic assumptions
		(H1)--(H2), whereas (P4) follows from strict Hugoniot convexity together
		with (P1).
		Property (P2) additionally requires the Hugoniot-convexity
		assumption~\textup{(H4b)}: the strict convexity of $\Si$ is precisely what
		ensures the profile remains on the correct (negative) side of the
		Hugoniot chord, without invoking the Lax entropy condition circularly.
		The key structural fact is (P3): the linear behavior of $\Si(F)$
		near $F=1$, with slope $\rhoo(s^2-c_S^2(\tho))$, follows exactly from
		thermodynamic compatibility and the upstream acoustic normalization, and
		holds for every admissible equation of state in the present class.
	\end{remark}

\begin{corollary}[Exclusion of nontrivial constant-temperature profiles]
\label{cor:no-constant-temperature}
Under Assumptions~\ref{ass:state} and
\ref{ass:hugoniot-convexity}, a nontrivial right-going compressive
traveling-wave profile of the full nonisothermal system cannot satisfy
$\Theta(F;s)\equiv\theta_0$ on $[F_-,1]$.
\end{corollary}

\begin{proof}
If $\Theta$ were constant, then $\Theta'=0$. Equation
\eqref{ThetaPrimeGeneral} would therefore give
\[
\Sigma(F;s)=\theta_0 p_\theta(F,\theta_0)\ge0
\]
throughout the profile. This contradicts
Theorem~\ref{thm:Sigma}\textup{(P2)}, which yields
$\Sigma(F;s)<0$ for every $F\in(F_-,1)$. Hence the constant-temperature
reduction is a principal subsystem at the level of the field equations, but
its nontrivial shock profiles are not solutions of the full energy-coupled
traveling-wave problem.
\end{proof}

\begin{corollary}[Smooth continuation of the downstream Hugoniot state]
\label{cor:HugoniotContinuation}
Let a selected compressive Hugoniot branch be defined for
$M_0\in(M_1,\Mst)$, with $M_1>1$, and suppose that, as
$M_0\nearrow\Mst$, its downstream states remain in a compact subset of the
admissible thermodynamic domain, with $F_-$ bounded away from zero and
$\theta_-$ bounded. Then $F_-(M_0)$ and $\theta_-(M_0)$ have unique limits
with $F_-^*:=F_-(\Mst)\in(0,1)$ and extend uniquely as $C^2$ functions of
$M_0$ to a neighborhood of $\Mst$. Hence no separate continuity assumption
on the Hugoniot branch is needed; the only possible obstruction is escape
from the admissible state domain before the critical Mach number is reached.
\end{corollary}

\begin{proof}
By \textup{(H2)}, $\Theta(F;s)$ is jointly $C^2$ in $(F,s)$, and therefore
$\Sigma(F;s)$ is jointly $C^2$. Along the branch,
$\Sigma(F_-(s);s)=0$, while Theorem~\ref{thm:Sigma}\textup{(P4)} gives
$\partial_F\Sigma(F_-(s);s)<0$. The implicit-function theorem thus yields a
unique local $C^2$ parametrization of the downstream state by $s$, and hence
by $M_0$.

The assumed compactness gives an accumulation point as
$M_0\nearrow\Mst$. Such a point cannot have $F_-^*=1$. Indeed, if
$F_-(M_0)\to1$, the secant formula~\eqref{secantspeed}, the mean-value
theorem, and $P_s'(1)=-\rho_0c_S^2(\theta_0)$ would imply
$s^2\to c_S^2(\theta_0)$, contradicting
$s^2\to(\Mst)^2c_S^2(\theta_0)>c_S^2(\theta_0)$. At the limiting value
$s_*=\Mst c_S(\theta_0)$, strict convexity of $\Sigma(\cdot;s_*)$ permits at
most one zero distinct from the persistent upstream zero $F=1$. Therefore
all convergent subsequences have the same nontrivial limit $F_-^*\in(0,1)$,
and the whole branch converges. Strict convexity also gives
$\partial_F\Sigma(F_-^*;s_*)<0$, so a final application of the
implicit-function theorem provides the unique $C^2$ continuation through
$M_0=\Mst$. Finally,
$\theta_-(M_0)=\Theta(F_-(M_0);M_0c_S(\theta_0))$ has the same regularity.
\end{proof}
	
\section{Main results: critical thresholds and constitutive-dependent amplitudes}
	\label{sec:main}
Throughout this section, $\Theta(F;s)$ denotes exclusively the temperature
selected by the energy integral~\eqref{energyint} along the traveling-wave
branch; primes on $\Theta$, $\Sigma$, and $\Pi$ denote differentiation with
respect to $F$.
\subsection{Critical Mach number and existence of the continuous profile}

	Ruggeri~\cite{SBS} first established that a smooth shock-structure
		solution can lose regularity only at a characteristic encounter, namely when
		the shock speed meets an eigenvalue of the governing hyperbolic system.
		An apparent pole may occasionally be cancelled by a simultaneous zero of the
		production numerator, so a local encounter alone is not always sufficient.
		The decisive global result is the theorem of Boillat and Ruggeri~\cite{Breakdown}:
		for a hyperbolic system of balance laws endowed with a convex entropy, a
		$C^1$ shock structure cannot exist when the shock speed exceeds the maximum
		characteristic velocity evaluated at the equilibrium state in front of the
		shock.  The principal-subsystem and subcharacteristic framework underlying
		this result was developed in~\cite{BoillatRuggeri1997}.

		The present model satisfies the entropy and symmetric-hyperbolicity hypotheses
		by the construction of Arima and Ruggeri~\cite{ArimaRuggeri2026}.  In the
		right-going one-dimensional reduction there is a single positive acoustic
		eigenvalue $\lambda_0$, hence it is the maximal characteristic velocity.
		Consequently the Boillat--Ruggeri threshold is exactly the upstream encounter
		$s=\lambda_0(1)$.  The identity~\eqref{Dcharacteristic} shows that this same
		condition is represented in the reduced ODE by $D(1;M_0)=0$. The
general theorem rules out a $C^1$ shock profile above the threshold;
consequently, any admissible piecewise-smooth continuation, if it exists,
must contain a subshock. Below the threshold, the global characteristic
ordering proved in Lemma~\ref{lem:Dpositive} excludes any earlier interior
encounter and establishes regularity of the complete continuous branch.

		The following lemma makes this precise: it evaluates the characteristic
		identity~\eqref{Dcharacteristic} at the upstream equilibrium state, uses
		the dimensionless viscosity parameter $\bmuo$ defined in~\eqref{bmu0def}
		to identify the critical Mach number $\Mst$, and gives the local behavior
		of $D$ near the critical point $(F,M_0)=(1,\Mst)$.

\begin{lemma}[Global characteristic ordering, critical Mach number, and local expansion]
\label{lem:Dpositive}
Under Assumptions~\ref{ass:state} and~\ref{ass:hugoniot-convexity}, the
positive nonequilibrium characteristic speed satisfies
\begin{equation}\label{globalCharacteristicOrdering}
\lambda_0^2(F;s)-\lambda_0^2(1;s)
=
\frac{P_s'(1)-P_s'(F)}{\rhoo}
+\frac{\Si(F;s)\Theta'(F;s)}{\rhoo\Theta(F;s)}
+\frac{\mu_0}{\rhoo\tau_0}
\left(\frac{\Theta(F;s)}{\theta_0}-1\right).
\end{equation}
Consequently,
\[
\lambda_0(F;s)>\lambda_0(1;s),\qquad F\in[F_-,1),
\]
and the upstream state is the strict global minimum of the positive
characteristic speed along every nontrivial compressive profile.
Moreover,
\[
D(F;M_0)>0\quad\Longleftrightarrow\quad s<\lambda_0(F;M_0),
\]
and
\[
D(1;M_0)=1-\frac{M_0^2-1}{\bmuo}.
\]
Hence the first characteristic encounter occurs at the upstream state and
the critical Mach number is
\begin{equation}\label{criticalMach}
\Mst=\sqrt{1+\bmuo}.
\end{equation}
In particular,
\[
D(F;M_0)>0,
\qquad 1<M_0<\Mst,
\qquad F\in[F_-(M_0),1].
\]
If the selected branch satisfies the compactness condition of
Corollary~\ref{cor:HugoniotContinuation}, that corollary supplies its unique
$C^2$ continuation to $M_0=\Mst$, and
\[
D(F;\Mst)>0
\quad\text{for }F\in[F_-(\Mst),1),
\qquad D(1;\Mst)=0.
\]
Finally, with $u=1-F$ and $\delta=\Mst-M_0$,
\begin{equation}\label{Dnearcrit}
D(1-u;M_0)=\dMach\,\delta+\dDef\,u+o(\delta+u),
\end{equation}
where
\begin{equation}\label{dMdFdef}
\dMach=-\partial_{M_0}D(1;\Mst)
=\frac{2\Mst}{\bmuo}>0,
\qquad
\dDef=-\partial_FD(1;\Mst)>0,
\end{equation}
and
\begin{equation}\label{dFexplicit}
\dDef
=\frac{\tau_0}{\mu_0}\Si''(1;\Mst)
+\frac{2p_\theta(1,\tho)}{\rhoo C_v(1,\tho)}.
\end{equation}
\end{lemma}

\begin{proof}
Using $P_s'=p_F+p_\theta\Theta'$ together with
\eqref{ThetaPrimeGeneral}, the characteristic speed \eqref{lambda0local}
can be rewritten exactly as
\begin{equation}\label{lambdaReducedPressureIdentity}
\lambda_0^2(F;s)
=-\frac{P_s'(F)}{\rhoo}
+\frac{\Si(F;s)\Theta'(F;s)}{\rhoo\Theta(F;s)}
+\frac{\mu_0\Theta(F;s)}{\rhoo\tau_0\theta_0}.
\end{equation}
For completeness, the algebraic equivalence with
\eqref{lambda0local} is displayed explicitly:
\begin{align*}
-\frac{P_s'}{\rhoo}
+\frac{\Si\Theta'}{\rhoo\Theta}
&=-\frac{p_F+p_\theta\Theta'}{\rhoo}
+\frac{\Si\Theta'}{\rhoo\Theta}\\
&=-\frac{p_F}{\rhoo}
+\frac{(\Si-\Theta p_\theta)^2}
{\rhoo^2\Theta\eps_\theta},
\end{align*}
where \eqref{ThetaPrimeGeneral} was used in the second line.
At $F=1$, $\Si(1)=0$ and $\Theta(1)=\theta_0$, and subtraction gives
\eqref{globalCharacteristicOrdering}.

Each term on the right-hand side is nonnegative. Indeed, strict Hugoniot
convexity gives $P_s'(F)<P_s'(1)$ for $F<1$. Theorem~\ref{thm:Sigma}(P2)
gives $\Si(F)<0$ in the interior, while
\[
\Theta'=\frac{\Si-\Theta p_\theta}{\rhoo\eps_\theta}<0
\]
by $p_\theta\ge0$ and $\eps_\theta>0$; hence $\Si\Theta'>0$. The same
inequality implies $\Theta(F)>\theta_0$ for $F<1$. The first term is
strictly positive in the interior, proving the global ordering. Notice that
the conclusion is a global minimum principle,
$\lambda_0(F;s)>\lambda_0(1;s)$ for $F<1$; it does not require, and does
not assert, the stronger pointwise monotonicity condition
$\partial_F\lambda_0<0$.

The characteristic identity \eqref{Dcharacteristic} shows that $D$ is a
positive multiple of $\lambda_0^2-s^2$. At the upstream state,
\[
\lambda_0^2(1;M_0)
=c_S^2(\tho)+\frac{\mu_0}{\rhoo\tau_0}
=c_S^2(\tho)(1+\bmuo),
\]
and $s=M_0c_S(\tho)$, which gives the formula for $D(1;M_0)$ and
\eqref{criticalMach}. The global ordering then proves all stated sign
properties of $D$.

Taylor expansion at $(F,M_0)=(1,\Mst)$ yields \eqref{Dnearcrit} and
$\dMach>0$. Since $\Pi=\Si/\Theta$, $\Si(1)=0$, and
$\Si'(1;\Mst)=\mu_0/\tau_0$, direct differentiation gives
\eqref{dFexplicit}. Finally $\Si''(1;\Mst)>0$ by Hugoniot convexity and
$p_\theta(1,\theta_0)\ge0$, so $\dDef>0$.
\end{proof}

\begin{proposition}[Boillat--Ruggeri loss-of-regularity threshold]
\label{prop:BRthreshold}
Let the hypotheses of the nonisothermal RET model and its convex entropy
hold.  If $M_0>\Mst$, then
\[
s=M_0c_S(\theta_0)>\lambda_0(1;M_0),
\]
where $\lambda_0(1;M_0)$ is the maximal characteristic speed at the
unperturbed equilibrium state.  Therefore a $C^1$ shock structure is
impossible.  Consequently, if an admissible piecewise-smooth connection between the two
equilibrium end states exists, it must contain at least one subshock.
\end{proposition}

\begin{proof}
At the upstream equilibrium state, Lemma~\ref{lem:Dpositive} gives
\[
\lambda_0^2(1)=c_S^2(\theta_0)(1+\bmuo)
=c_S^2(\theta_0)(\Mst)^2.
\]
Thus $M_0>\Mst$ is equivalent to $s>\lambda_0(1)$.  Since the system is
symmetric hyperbolic and possesses a convex entropy~\cite{ArimaRuggeri2026},
the theorem of Boillat and Ruggeri~\cite{Breakdown} applies and excludes a
continuously differentiable shock structure.  Therefore any admissible piecewise-smooth continuation, if it exists,
must contain a discontinuity.
\end{proof}

\begin{theorem}[Existence of the monotone traveling-wave profile]
		\label{thm:existence}
		Consider a selected right-going compressive Hugoniot branch satisfying
		Assumptions~\ref{ass:state}, \ref{ass:k}, and~\ref{ass:hugoniot-convexity},
		and let $\bmuo>0$. For a Mach number $1<M_0<\Mst$ on that branch,
		with $\Mst$ given by \eqref{criticalMach}, there exists a monotone
		$C^1$ solution $F(\varphi)$ of the ODE \eqref{ODE}, strictly increasing
		from the Rankine--Hugoniot downstream state $F_-$ to the upstream state
		$1$, defined on an interval $I=(\varphi_-,\varphi_+)$ with:
		\begin{itemize}
			\item $\varphi_- = -\infty, \quad \varphi_+ = +\infty$ \quad for $0 < m \le 1$,
			\item $\varphi_-$ and $\varphi_+$ finite \quad for $m > 1$.
		\end{itemize}
		For $m > 1$, the compact profile can be extended to all of $\mathbb{R}$ by constant extrapolation.
	\end{theorem}

	\begin{proof}
		By Lemma~\ref{lem:Dpositive}, the denominator $D(F;M_0)$ is
		strictly positive on $[F_-,1]$ for $1<M_0<\Mst$. By
		Theorem~\ref{thm:Sigma}\,\textup{(P2)}, $\Si(F)<0$ on $(F_-,1)$, and
		therefore
		\[
		\calN(F)=-F\,2^{1/m-1}k(\Th(F))^{-1/m}
		|\Si(F)|^{(1-m)/m}\Si(F)>0.
		\]
		Fix $F_c\in(F_-,1)$ and define the quadrature
		\begin{equation}\label{profileQuadrature}
		\varphi(F):=s\int_{F_c}^{F}\frac{D(r)}{\calN(r)}\,dr.
		\end{equation}
		It is strictly increasing and hence invertible. Its inverse is a $C^1$
		solution of \eqref{ODE}, with $F_\varphi>0$ in the open interval.

		The endpoint behavior follows from the nonzero endpoint slopes. Indeed,
		Theorem~\ref{thm:Sigma}\,\textup{(P3)}--\textup{(P5)} gives
		$|\Si(F)|\asymp|F-F_e|$ near either equilibrium endpoint
		$F_e\in\{F_-,1\}$. Since $k(\Th)$ is bounded above and below and
		$D$ is positive and bounded away from zero on the compact subcritical
		branch,
		\[
		\frac{sD(F)}{\calN(F)}\asymp |F-F_e|^{-1/m}.
		\]
		Thus the integral \eqref{profileQuadrature} diverges at an endpoint if
		and only if $0<m\le1$. In this case its range is all of $\mathbb R$.
		For $m>1$ both endpoint integrals are finite; moreover
		$F_\varphi\asymp|F-F_e|^{1/m}\to0$, so constant extrapolation beyond
		the two finite endpoints produces a global $C^1$ profile.
	\end{proof}

	\subsection{Definition of the dimensionless shock thickness}
The traveling-wave coordinate $\varphi=X-st$ has the dimensions of
length. We introduce the natural RET relaxation length
\begin{equation}\label{naturalLength}
\ell_0:=\tau_0c_0,
\qquad c_0:=c_S(\theta_0),
\end{equation}
together with the dimensionless traveling-wave coordinate
\begin{equation}\label{xidef}
\xi:=\frac{\varphi}{\ell_0}.
\end{equation}
From the mass integral~\eqref{mass}, $v=s(1-F)$, so the velocity varies
monotonically from $0$ upstream to $v_-=s(1-F_-)$ downstream. Introducing
\[
w(\xi)=\frac{v(\xi)}{v_-}
=\frac{1-F(\xi)}{1-F_-}\in[0,1],
\]
we define directly the dimensionless maximum-slope shock thickness by
\begin{equation}\label{DeltaDef}
\Delta:=\frac{1-F_-}{\max_{\xi}|F_\xi(\xi)|}
=\frac{1-F_-}{\ell_0\max_{\varphi}|F_\varphi(\varphi)|}.
\end{equation}
Equivalently, if
\[
\Delta_{\mathrm{phys}}
:=\frac{1-F_-}{\max_{\varphi}|F_\varphi(\varphi)|}
\]
denotes the dimensional (physical) shock thickness, then
\begin{equation}\label{DeltaPhysicalRelation}
\Delta=\frac{\Delta_{\mathrm{phys}}}{\ell_0}.
\end{equation}
Thus $\Delta$ measures the physical shock thickness in units of the
natural relaxation length $\ell_0$. Since $\ell_0$
is fixed for a prescribed upstream state, dimensional and dimensionless
thicknesses have the same critical exponents.

\begin{lemma}[Weak-shock Hugoniot parametrization]
			\label{lem:weakHugoniot}
			Set
			\begin{equation*}
				\delta_F=1-F_-.
			\end{equation*}
			Along the right-going compressive Hugoniot branch, the weak-shock limit
			is equivalently described by
			\begin{equation*}
				\delta_F\to0^+,
				\qquad
				M_0\to1^+,
			\end{equation*}
			and one has
			\begin{equation}
			\label{weakHugoniotParam}
				\delta_F=A(M_0-1)+o(M_0-1),
				\qquad A>0.
			\end{equation}
			More precisely, if
			\begin{equation*}
				\mathcal P(\delta_F):=p(1-\delta_F,\theta_-(\delta_F)),
			\end{equation*}
			where $\theta_-(\delta_F)$ is determined by the thermodynamic Hugoniot
			relation~\eqref{RHthermal}, then
			\begin{equation*}
				A=\frac{4\rho_0 c_S^2(\theta_0)}{\mathcal P''(0)}.
			\end{equation*}
		\end{lemma}

		\begin{proof}
			By Assumption~\ref{ass:state}, \eqref{RHthermal} defines a local
			$C^2$ downstream temperature branch
			$\theta_- = \theta_-(\delta_F)$, with $\theta_-(0)=\theta_0$. Differentiating
			\eqref{RHthermal} at $\delta_F=0$ gives
			\begin{equation*}
				\theta_-'(0)
				=\frac{\theta_0p_\theta(1,\theta_0)}{\rho_0C_v(1,\theta_0)}.
			\end{equation*}
			Hence
			\begin{align*}
				\mathcal P'(0)
				&=-p_F(1,\theta_0)
				+p_\theta(1,\theta_0)\theta_-'(0)\\
				&=\rho_0c_S^2(\theta_0),
			\end{align*}
			by the definition of the adiabatic sound speed~\eqref{adiabaticSound}.
			The momentum Rankine--Hugoniot relation gives
			\begin{equation*}
				s^2(\delta_F)=\frac{\mathcal P(\delta_F)-p_0}{\rho_0\delta_F}.
			\end{equation*}
			Therefore Taylor expansion at $\delta_F=0$ yields
			\begin{equation*}
				s^2(\delta_F)
				=c_S^2(\theta_0)+\frac{\mathcal P''(0)}{2\rho_0}\delta_F+o(\delta_F).
			\end{equation*}
			By genuine nonlinearity with the classical orientation in
			\textup{(H4a)}, $\mathcal P''(0)>0$ on the selected compressive branch. Using \eqref{MachDef}, we obtain
			\begin{equation*}
				M_0-1
				=\frac{\mathcal P''(0)}{4\rho_0c_S^2(\theta_0)}\delta_F+o(\delta_F),
			\end{equation*}
			which can be inverted and gives \eqref{weakHugoniotParam}.
		\end{proof}

	\subsection{Weak-shock regime}
	
	\begin{theorem}[Universal threshold $m=2$ and constitutive-dependent amplitude]\label{thm:weak}
Let Assumptions~\ref{ass:state}, \ref{ass:k}, and~\ref{ass:hugoniot-convexity} hold.
For $M_0>1$ close to $1$, the shock thickness satisfies
\begin{equation}\label{weakscaling}
\Delta \sim C_{\mathrm w}(m,p,\eps,k,\tau_0)
(M_0-1)^{1-2/m}
\quad\text{as }M_0\to1^+,
\end{equation}
where, with $A$ defined in \eqref{weakHugoniotParam},
\begin{equation}\label{CweakExplicit}
C_{\mathrm w}
=\frac{2^{1+2/m}}{\tau_0}
\left(\frac{k(\theta_0)}{\mathcal P''(0)}\right)^{1/m}
A^{1-2/m}>0.
\end{equation}
Hence the exponent $1-2/m$ is universal, while the amplitude is not.  In
particular:
\begin{itemize}
\item $m<2$: $\Delta\to+\infty$;
\item $m=2$: $\Delta\to\Delta_*>0$, with
\begin{equation}\label{weakPlateauExplicit}
\Delta_*=\frac{4}{\tau_0}
\left(\frac{k(\theta_0)}{\mathcal P''(0)}\right)^{1/2};
\end{equation}
\item $m>2$: $\Delta\to0$.
\end{itemize}
\end{theorem}

\begin{proof}
Set $\delta_M=M_0-1$ and retain the deformation jump
$\delta_F=1-F_-$ of Lemma~\ref{lem:weakHugoniot}.  Then
\begin{equation}\label{weakjump}
\delta_F=A\delta_M+o(\delta_M),\qquad A>0.
\end{equation}
For $z=(1-F)/\delta_F\in[0,1]$, introduce
\[
\widehat\Sigma_{\delta_F}(z):=\delta_F^{-2}\Si(1-\delta_F z;s(\delta_F)).
\]
Both endpoint values vanish.  Moreover,
\[
\widehat\Sigma_{\delta_F}''(z)=\Si''(1-\delta_F z;s(\delta_F)).
\]
The weak-Hugoniot expansion gives
\[
\Si'(1;s(\delta_F))
=\rho_0\bigl(s^2(\delta_F)-c_S^2(\theta_0)\bigr)
=\frac{\mathcal P''(0)}{2}\,\delta_F+o(\delta_F).
\]
The second-order curvature is identified most directly from the reduced
pressure. Since $P_s(1)=p_0$ and
$P_s'(1)=-\rho_0c_S^2(\theta_0)$ for every $s$, Taylor expansion, uniform
for $s$ near $c_S(\theta_0)$, gives
\[
\mathcal P(\delta_F)
=P_{s(\delta_F)}(1-\delta_F)
=p_0+\rho_0c_S^2(\theta_0)\delta_F
+\frac12P_{s(\delta_F)}''(1)\delta_F^2+o(\delta_F^2).
\]
Because $s(\delta_F)\to c_S(\theta_0)$ and $P_s''$ is continuous in
$(F,s)$ by \textup{(H2)},
\[
\mathcal P''(0)=P_{c_S(\theta_0)}''(1)
=\Si''(1;c_S(\theta_0)),
\]
the last equality following because the Rankine--Hugoniot chord is affine.
Since the constitutive branch is $C^2$,
$\widehat\Sigma_{\delta_F}''$ converges uniformly to
$\mathcal P''(0)$ on $[0,1]$.  Together with the two zero boundary values,
this yields the uniform limit
\begin{equation}\label{SigmaWeakUniform}
\delta_F^{-2}\Si(1-\delta_F z;s(\delta_F))
\longrightarrow
-\frac{\mathcal P''(0)}{2}z(1-z)
\quad\text{uniformly on }[0,1].
\end{equation}
This is stronger than a two-sided order estimate and is the point that
justifies the use of the asymptotic symbol $\sim$.

In the same limit, $F\to1$, $s\to c_0$, $\Theta\to\theta_0$,
$k(\Theta)\to k(\theta_0)$, and $D\to1$ uniformly.  From
\eqref{ODE}, \eqref{numerator}, and \eqref{SigmaWeakUniform},
\begin{equation}\label{FprimeWeakUniform}
\delta_F^{-2/m}F_\varphi(1-\delta_F z)
\longrightarrow
G_{\mathrm w}(z):=
\frac{\mathcal P''(0)^{1/m}}
{2c_0k(\theta_0)^{1/m}}
[z(1-z)]^{1/m}
\end{equation}
uniformly on $[0,1]$.  The limiting function has a unique maximum at
$z=1/2$; therefore the maxima converge:
\[
\delta_F^{-2/m}\max|F_\varphi|
\longrightarrow
\frac{\mathcal P''(0)^{1/m}}
{2^{1+2/m}c_0k(\theta_0)^{1/m}}.
\]
Substitution into \eqref{DeltaDef}, using $\ell_0=\tau_0c_0$ and
\eqref{weakjump}, gives \eqref{weakscaling}--\eqref{CweakExplicit}.
The three regimes and the plateau \eqref{weakPlateauExplicit} follow
immediately.
\end{proof}

\begin{remark}
		The key step is the uniform limit \eqref{SigmaWeakUniform}: the linear behavior of $\Si(F)$
		near $F=1$, with slope proportional to $(M_0^2-1)$, is an exact
		consequence of thermodynamic compatibility and the finite upstream sound
		speed, and holds for every admissible $p(F,\theta)$ in the present class.
		The thermal coupling through $\Th(F)$ enters
		only at order $O((1-F)^2)$ and does not affect the leading-order exponent.
		In particular, the exponent $1-2/m$ is the same as in the isothermal
		case~\cite{RuggeriTaniguchi}: the nonisothermal generalization does not change the
		universality class.
	\end{remark}
	
	\subsection{Near-critical regime}
\begin{proposition}[Localization of the upstream supercritical singularity]
\label{prop:subshock}
Under Assumptions~\ref{ass:state}--\ref{ass:hugoniot-convexity}, suppose
the selected equilibrium Hugoniot branch satisfies the compactness condition
of Corollary~\ref{cor:HugoniotContinuation}. Let
$F_-^*:=F_-(\Mst)\in(0,1)$ denote its unique $C^2$ continuation. There exist
$\eta>0$ and a
neighborhood of $\Mst$ such that, for $M_0>\Mst$ sufficiently close to
$\Mst$, there is a unique value
\[
F^*(M_0)\in(1-\eta,1)\cap(F_-(M_0),1)
\]
satisfying
\[
D(F^*;M_0)=0.
\]
Moreover, $\calN(F^*)>0$. Hence the zero of the denominator is not
cancelled by the numerator, and the continuous traveling-wave branch
cannot pass through $F^*$.  For slightly supercritical shocks, this local pole is the reduced-ODE
manifestation of the loss of $C^1$ regularity established globally by
Proposition~\ref{prop:BRthreshold}.
\end{proposition}

\begin{proof}
Let $u=1-F$ and $\delta_+=M_0-\Mst>0$. The expansion
\eqref{Dnearcrit}, written on the supercritical side, gives
\[
D(1-u;M_0)
=-\dMach\,\delta_+ +\dDef\,u+o(\delta_++u),
\qquad \dMach,\dDef>0.
\]
The leading equation has the unique zero
\[
u^*=\frac{\dMach}{\dDef}\,\delta_++o(\delta_+).
\]
The implicit-function theorem therefore gives a unique zero of $D$ in a
fixed neighborhood of $F=1$. No global uniqueness claim for possible
additional zeros deeper in a supercritical profile is needed here. Since $F_-(\Mst)<1$ and the Hugoniot branch is continuous, one has
$F^*(M_0)=1-u^*\in(F_-(M_0),1)$ for $\delta_+$ sufficiently small.
Theorem~\ref{thm:Sigma}(P2) then gives $\Si(F^*)<0$, and the explicit
formula~\eqref{numerator} implies $\calN(F^*)>0$. Thus the pole is
uncancelled and represents a genuine loss of the continuous profile.
\end{proof}

\begin{theorem}[Universal threshold $m=1$ and constitutive-dependent amplitudes]
\label{t:nearcrit}
Let Assumptions~\ref{ass:state}, \ref{ass:k}, and
\ref{ass:hugoniot-convexity} hold, and suppose that the selected Hugoniot
branch satisfies the compactness condition of
Corollary~\ref{cor:HugoniotContinuation}. Let
$F_-^*:=F_-(\Mst)\in(0,1)$ denote its uniquely continued downstream state.
For $M_0<\Mst$, the behavior of
$\Delta$ as $M_0\nearrow\Mst$ is as follows:
\begin{itemize}
\item For $0<m\le1$,
\[
\lim_{M_0\nearrow\Mst}\Delta
=\Delta_c(m,\bmuo,p,\eps,k)>0,
\]
and the limiting value is generally not universal.
\item For $m>1$,
\[
\Delta\sim C_c(m,\bmuo,p,\eps,k)
(\Mst-M_0)^{(m-1)/m},
\qquad C_c>0.
\]
\end{itemize}
Thus the threshold $m=1$ and the exponent $(m-1)/m$ are independent
of the constitutive functions, whereas finite limiting values and
amplitudes depend on the equation of state and on the temperature profile.
\end{theorem}

\begin{proof}
Let $\delta=\Mst-M_0>0$, $u=1-F$, and
$s_*=\Mst c_S(\theta_0)$.  Lemma~\ref{lem:Dpositive} gives, uniformly
for $(u,\delta)$ in a sufficiently small neighborhood of the origin,
\[
D(1-u;M_0)=\dMach\delta+\dDef u+o(\delta+u),
\qquad \dMach,\dDef>0.
\]
At criticality,
$|\Si(1-u)|=(\mu_0/\tau_0)u[1+o(1)]$ and
$k(\Theta(1-u))=k(\theta_0)[1+o(1)]$.  Hence, for $u=\delta y$,
\begin{equation}\label{NearCriticalScaledLimit}
\delta^{1-1/m}F_\varphi(1-\delta y;M_0)
\longrightarrow
G_*(y):=
\frac{B_*}{s_*}
\frac{y^{1/m}}{\dMach+\dDef y},
\end{equation}
uniformly for $y$ in every bounded interval, where
\begin{equation}\label{BstarDef}
B_*:=2^{1/m-1}
\left(\frac{\mu_0}{\tau_0k(\theta_0)}\right)^{1/m}>0.
\end{equation}

Suppose first that $m>1$.  The function $G_*$ has a unique maximum at
\[
y_* = \frac{\dMach}{\dDef(m-1)}.
\]
To see that this local maximum is also the global one, split the profile
into three regions.  On $u\ge\varepsilon$, the strict critical inequality in
Lemma~\ref{lem:Dpositive}, compactness, and continuity give a uniform
positive lower bound for $D$, so $F_\varphi$ is uniformly
bounded and becomes negligible after multiplication by
$\delta^{1-1/m}$.  In the intermediate region
$L\delta\le u\le\varepsilon$, the two-sided local estimates
$\calN\asymp u^{1/m}$ and $D\asymp\delta+u$ give
\[
\delta^{1-1/m}|F_\varphi|
\le C L^{-(m-1)/m},
\]
which is arbitrarily small for large $L$.  On $0\le u\le L\delta$,
\eqref{NearCriticalScaledLimit} is uniform.  It follows that
\[
\delta^{1-1/m}\max|F_\varphi|
\longrightarrow \max_{y\ge0}G_*(y)>0.
\]
The maximum can be evaluated explicitly:
\begin{equation}\label{GstarMaximum}
\max_{y\ge0}G_*(y)
=\frac{B_*}{s_*}
\frac{(m-1)^{(m-1)/m}}
{m\,\dMach^{(m-1)/m}\dDef^{1/m}}.
\end{equation}
Since $1-F_-(M_0)\to1-F_-^*>0$,
\begin{equation}\label{CcExplicit}
C_c=
\frac{(1-F_-^*)s_*}{\ell_0B_*}
\frac{m\,\dMach^{(m-1)/m}\dDef^{1/m}}
{(m-1)^{(m-1)/m}}>0,
\end{equation}
and the asserted asymptotic equivalence follows.  Thus the symbol
$\sim$ is justified by convergence of the rescaled global maximum, not
merely by upper and lower bounds.

For $0<m<1$, the critical local slope behaves as
$u^{1/m-1}\to0$ at the upstream endpoint. Moreover, for sufficiently small
$u$ and $\delta$ the local estimate
\[
|F_\varphi(1-u;M_0)|
\le C\frac{u^{1/m}}{\delta+u}
\le C u^{1/m-1}
\]
is uniform. Hence the contribution of a small upstream neighborhood to the
maximum is uniformly small, while away from that neighborhood the ODE
right-hand side converges uniformly. The maximum slope therefore converges
to a finite, strictly positive maximum of the critical profile.

For $m=1$, the convergence is not uniform exactly at $u=0$, but the local
boundary-layer quotient is bounded and satisfies
\[
F_\varphi(1-u;M_0)
=\frac{B_*}{s_*}\frac{u}{\dMach\delta+\dDef u}\,[1+o(1)].
\]
The local supremum is bounded above by $B_*/(s_*\dDef)+o(1)$, while
choosing, for example, $u=\sqrt{\delta}$ gives the matching lower limit.
Thus the local supremum tends to $B_*/(s_*\dDef)$, which is precisely the
finite one-sided critical slope.
Combining the convergence of the local endpoint supremum with uniform
convergence on compact subsets away from $F=1$, the global maximum converges
to the larger of the critical interior maximum and the one-sided endpoint
value $B_*/(s_*\dDef)$. In particular, it converges to a finite nonzero
value. Consequently $\Delta\to\Delta_c>0$ for all $0<m\le1$.
\end{proof}

\section{Physical constitutive functions for non-Newtonian fluids}
\label{sec:constitutive}

We now discuss the most physically relevant constitutive choices
for non-Newtonian fluids, verify the thermodynamic and rheological
hypotheses (H1)--(H3), verify Hugoniot convexity, and illustrate the global characteristic-ordering result.

\subsection{Equation of state for nearly incompressible fluids}

Most non-Newtonian fluids of practical interest (polymer gels,
biological soft tissues, dense suspensions of cornstarch)
are nearly incompressible.
The standard model is the
	\textbf{Tait--Murnaghan equation of state}~\cite{Tait1888,Murnaghan1944,DymondMalhotra1988}:
\begin{equation}\label{Tait}
p(\rho,\theta)
= p_\infty\!\left[\left(\frac{\rho}{\rhoo}\right)^n - 1\right]
+ a_T(\theta-\tho),
\end{equation}
where $p_\infty>0$ is a reference pressure, $n\ge1$ is the Tait exponent
(typically $n\approx7$ for water-like fluids), and $a_T\ge0$ is a constant
isochoric thermal-pressure coefficient. At a chosen reference state one may
write $a_T=\alpha_0K_{T0}$ in terms of the reference volumetric expansion
coefficient and isothermal bulk modulus; away from that state the actual
expansion coefficient is $\alpha=a_T/K_T$. In terms of $F=\rhoo/\rho$,
\begin{equation*}
p(F,\theta)=p_\infty\bigl[F^{-n}-1\bigr]+a_T(\theta-\tho).
\end{equation*}
The isothermal sound speed is
\begin{equation*}
c_T^2 = \frac{n\,p_\infty}{\rhoo}\,F^{1-n},
\end{equation*}
which is finite and positive for $F\in(0,\infty)$, satisfying the sound-speed part of (H1).
For $n=1$ and a small reference expansion ratio
$a_T/K_{T0}$, this form is consistent with the EQTI specialization
discussed by Gouin and Ruggeri~\cite{GouinRuggeri2012}.

\subsection{Internal energy}

From the Gibbs relation and Maxwell's identity, the internal energy
compatible with \eqref{Tait} is:
\begin{align}\label{eps_Tait}
\eps(F,\theta)-\eps(1,\tho)
&= C_v\,(\theta-\tho)
+ \frac{a_T\tho+p_\infty}{\rhoo}(F-1)\nonumber\\
&\quad + \frac{p_\infty}{\rhoo(n-1)}\left(F^{1-n}-1\right),
\end{align}
where $C_v>0$ satisfies the thermodynamic stability condition (H1) and $n \neq 1$.
For $n=1$ the last term is replaced by $-(p_\infty/\rhoo)\log F$,
which is also the continuous limit as $n\to1$.

\subsection{Temperature-dependent consistency coefficient}

Three physically motivated choices for $k(\theta)$, all satisfying
Assumption~\ref{ass:k}, are:

\begin{enumerate}
\item \textbf{Temperature-independent consistency coefficient}:
\begin{equation*}
k(\theta) \equiv k_0 > 0.
\end{equation*}
This choice removes the explicit thermal dependence of the rheological
coefficient, but the model remains nonisothermal because the energy balance
and the temperature profile are retained. Thus a constant consistency
coefficient must not be confused with the isothermal model studied in the
companion paper, in which the temperature is fixed and the energy balance is
omitted.
\item \textbf{Power-law thermal softening} (polymer melts):
\begin{equation*}
k(\theta) = k_0\left(\frac{\theta}{\tho}\right)^{-n_T},
\qquad n_T > 0.
\end{equation*}
Viscosity decreases with temperature (thermal thinning), typical
of polymer solutions and biological gels.
\item \textbf{Arrhenius model} (dense suspensions, food gels):
\begin{equation}\label{kArrhenius}
k(\theta) = k_0\exp\!\left(\frac{E_a}{R\,\theta}\right),
\end{equation}
where $E_a>0$ is the activation energy.
This is a common thermally activated model used to describe strong temperature sensitivity in complex fluids and soft materials; see, for example, the rheological contexts discussed in~\cite{BrownJaeger2014,Coussot2005}.
\end{enumerate}

All three choices satisfy $0<k_{\min}\le k(\theta)\le k_{\max}<\infty$
on any compact temperature interval $[\tho,\theta_-]$, confirming
Assumption~\ref{ass:k}.

\subsection{Verification of the temperature profile and Hugoniot convexity}

For the notation above, since the reference pressure is
$p_0=p(1,\theta_0)=0$, the energy integral can be solved explicitly:
\begin{equation}\label{ThetaTaitExplicit}
\Theta(F;s)=\theta_0+\frac{1}{C_v}
\left\{
\frac{s^2}{2}(1-F)^2
+\frac{a_T\theta_0+p_\infty}{\rho_0}(1-F)
-\frac{p_\infty}{\rho_0(n-1)}\bigl(F^{1-n}-1\bigr)
\right\},
\end{equation}
where, for $n=1$, the quotient $(F^{1-n}-1)/(n-1)$ is interpreted as
$-\log F$. This formula proves uniqueness and $C^2$ regularity, but the range
condition in \textup{(H2)} also requires monotonicity.  Put
$u=1-F\in[0,\delta_F]$, $\delta_F=1-F_-$, and define
\[
H(u):=C_v\frac{d\Theta}{du}
=s^2u+\frac{a_T\theta_0}{\rho_0}
-\frac{p_\infty}{\rho_0}
\left[(1-u)^{-n}-1\right].
\]
At the endpoints,
\[
H(0)=\frac{a_T\theta_0}{\rho_0}\ge0,
\qquad
H(\delta_F)=\frac{a_T\theta_-}{\rho_0}\ge0,
\]
where the second identity follows from the momentum
Rankine--Hugoniot relation.  Moreover,
\[
H''(u)=-\frac{n(n+1)p_\infty}{\rho_0}(1-u)^{-n-2}<0.
\]
Thus $H$ is strictly concave and lies above the chord joining its
non-negative endpoint values.  Hence $d\Theta/du\ge0$ on $[0,\delta_F]$
(strictly in the interior), and therefore
\[
\theta_0\le\Theta(F;s)\le\theta_-,
\qquad F\in[F_-,1].
\]
This supplies the range statement missing from mere explicit solvability.

The Hugoniot curvature can also be evaluated explicitly.  From
\eqref{ThetaTaitExplicit},
\begin{equation}\label{PsTaitSecond}
P_s''(F)=
np_\infty F^{-n-2}
\left[(n+1)-\frac{a_TF}{\rho_0C_v}\right]
+\frac{a_Ts^2}{C_v}.
\end{equation}
Consequently the simple sufficient condition
\begin{equation}\label{TaitConvexCondition}
\frac{a_T}{\rho_0C_v}\le n+1
\end{equation}
implies $P_s''(F)>0$ throughout $[F_-,1]$ and verifies
\textup{(H4b)}.

\begin{corollary}[Tait--Murnaghan fluid]\label{cor:Tait}
Let
\[
p(F,\theta)=p_\infty(F^{-n}-1)+a_T(\theta-\theta_0),
\qquad p_\infty,C_v>0,\quad n\ge1,\quad a_T\ge0.
\]
Then \textup{(H1)}--\textup{(H2)} hold on every positive-temperature
compressive branch.  More precisely:
\begin{enumerate}
\item the adiabatic sound speed is
\begin{equation}\label{TaitSoundSpeedFull}
c_S^2(F,\theta)=
\frac{np_\infty}{\rho_0}F^{1-n}
+\frac{\theta a_T^2}{\rho_0^2C_v}F^2>0;
\end{equation}
\item the exact temperature profile is given by
\eqref{ThetaTaitExplicit} and is monotone between $\theta_0$ and
$\theta_-$;
\item if \eqref{TaitConvexCondition} holds, then \textup{(H4b)} is
satisfied;
\item if \eqref{TaitConvexCondition} holds, then the global characteristic
ordering \eqref{globalCharacteristicOrdering} applies, so the upstream
state is the strict minimum of $\lambda_0$ and no earlier interior zero of
$D$ can occur; moreover,
\[
\dDef=\frac{\tau_0}{\mu_0}\Si''(1;\Mst)
+\frac{2a_T}{\rho_0C_v}>0.
\]
Thus both the no-earlier-encounter property and transversality are automatic.
\item on every bounded Mach-number interval, each positive-temperature
compressive branch satisfies the compactness condition of
Corollary~\ref{cor:HugoniotContinuation}. Consequently, whenever such a branch is defined for
$M_0\in(M_1,\Mst)$, Corollary~\ref{cor:HugoniotContinuation} applies and
its downstream state extends uniquely and smoothly through the critical
Mach number.
\end{enumerate}
\end{corollary}

\begin{proof}
The Gibbs relation follows from the internal energy in
\eqref{eps_Tait}, for which $\eps_\theta=C_v>0$,
$\eps_F=(\theta p_\theta-p)/\rho_0$, and
$p_F=-np_\infty F^{-n-1}<0$.  Formula
\eqref{TaitSoundSpeedFull} follows directly from
\eqref{adiabaticSound}.  The explicit solution
\eqref{ThetaTaitExplicit} and the concavity argument above prove all
parts of \textup{(H2)}, including the global temperature range.
Equation \eqref{PsTaitSecond} proves \textup{(H4b)} under
\eqref{TaitConvexCondition}.  Finally $p_\theta=a_T\ge0$, so
\eqref{dFexplicit} yields $\dDef>0$.

It remains to verify the compactness required in
Corollary~\ref{cor:HugoniotContinuation}. Let $s$ range in any bounded
interval. From \eqref{ThetaTaitExplicit},
\[
\Theta(F;s)\longrightarrow-\infty\qquad\text{as }F\downarrow0,
\]
uniformly for bounded $s$: for $n>1$ the divergent term is proportional to
$-F^{1-n}$, while for $n=1$ it is proportional to $\log F$. Since a
positive-temperature compressive branch satisfies
$\theta_-=\Theta(F_-;s)\ge\theta_0>0$, its downstream deformation is
therefore uniformly bounded away from zero. Once $F_-$ is confined to a
compact subinterval of $(0,1)$, the explicit formula
\eqref{ThetaTaitExplicit} also gives a uniform upper bound for $\theta_-$.
Thus the downstream states remain in a compact subset of the admissible
thermodynamic domain, as claimed.
\end{proof}

\section{Numerical setup and normalization}
\label{sec:numerical}

To perform the numerical integration of the shock-wave profiles and ensure computational stability across different rheological regimes, the governing equations have been cast into a dimensionless form.

\subsection{Dimensionless variables and normalization}
The dimensionless traveling-wave coordinate $\xi$ and shock thickness
$\Delta$ have already been introduced in Eqs.~\eqref{xidef}--\eqref{DeltaPhysicalRelation}.
Stress is normalized by the upstream dynamic pressure $\rho_0c_0^2$.
The deformation gradient remains dimensionless and connects
$F(+\infty)=1$ to $F(-\infty)=F_-$. Accordingly, the symbol $\Delta$ used in the figures has throughout
the dimensionless meaning fixed in \eqref{DeltaDef}--\eqref{DeltaPhysicalRelation}.

\subsection{Material parameters and reproducibility}
The base computations in Figs.~\ref{fig:shock_profiles} and
\ref{fig:shock_thickness_3panels_final} use the completely specified
normalized set
\begin{equation}\label{NumericalParameters}
\rho_0=\theta_0=C_v=K_0=k_0=\tau_0=c_0=1,
\quad n=5,
\quad \alpha_T=0.02,
\quad a_T:=\alpha_T K_0=0.02,
\quad \gamma=0.5.
\end{equation}
Here $K_0$ and $\alpha_T$ are used only to parameterize the thermal-pressure
coefficient $a_T=\alpha_T K_0$; the constitutive equations depend on $a_T$.
The value
\[
p_\infty=\frac{1-a_T^2/C_v}{n}=0.19992
\]
is chosen so that \eqref{TaitSoundSpeedFull} gives exactly
$c_S(\theta_0)=c_0=1$.  The consistency law is
\[
k(\Theta)=k_0\exp[-\gamma(\Theta-\theta_0)].
\]
For a prescribed $\bmuo$, we set
$\mu_0=\bmuo\tau_0\rho_0c_0^2$; hence, with the present normalization,
$\mu_0=\bmuo$ and $\omega=1/\bmuo$.  Figure~\ref{fig:shock_profiles}
uses $M_0=1.4$ and $\bmuo=2.0$, whereas
Fig.~\ref{fig:shock_thickness_3panels_final} uses
$\bmuo=0.10$ and $0.20$ and the values of $m$ reported in its legends.
Condition \eqref{TaitConvexCondition} is strongly satisfied because
$a_T/(\rho_0C_v)=0.02<6$.

For each $M_0$, the exact temperature formula
\eqref{ThetaTaitExplicit} is substituted into $\Si(F;s)$.  The nontrivial
Rankine--Hugoniot root $F_-\in(0,1)$ is found by Brent's method with
relative and absolute tolerance $10^{-13}$.  The maximum of $|F_\xi|$
is evaluated on a composite grid of 3600 points in
$z=(1-F)/(1-F_-)$, clustered at both endpoints.  The Mach-number curves
use 260 points, logarithmically clustered near $M_0=1$ and $M_0=\Mst$.
A grid-refinement check was performed by increasing the composite grid from
3600 to 7200 points for all 3201 parameter--Mach-number cases used in
Figs.~\ref{fig:shock_profiles}--\ref{fig:thermal_effect}. The maximum
relative change in $\Delta$ was $1.57\times10^{-5}$ (about $0.0016\%$).
The profiles in Fig.~\ref{fig:shock_profiles} are integrated with an
adaptive Runge--Kutta solver using relative tolerance $2\times10^{-10}$
and absolute tolerance $2\times10^{-12}$.

As a numerical validation of Lemma~\ref{lem:Dpositive}, the denominator
$D$ was also sampled at 4000 deformation values on every plotted profile.
For the parameter sets used in the figures, the sampled values decrease
toward the upstream state, where the sampled minimum occurs, and remain
positive on every subcritical branch. This observed monotonic behavior of
$D$ is a property of the displayed parameter sets; it is not asserted by
Lemma~\ref{lem:Dpositive}, which proves the global minimum property of
$\lambda_0$ and the positivity of $D$. On meshes ending
$2\times10^{-5}$ below the theoretical critical Mach number, the smallest
sampled values are $4.2\times10^{-4}$ for $\bmuo=0.10$ and
$2.2\times10^{-4}$ for $\bmuo=0.20$; for
Fig.~\ref{fig:shock_profiles}, $D(1;1.4)=0.52$. These computations confirm
the analytic characteristic ordering and are not used as a substitute for
it.

Figure~\ref{fig:shock_profiles} displays representative profiles for several values of the flow index $m$ at fixed $M_0=1.4$ and $\bmuo=2.0$.  This value of $\bmuo$ gives $\Mst=\sqrt{3}$, so all the profiles are strictly subcritical.  Since $M_0$ and the equilibrium equation of state are fixed, all curves connect the same two Rankine--Hugoniot equilibrium states,
\[
        F(+\infty)=1,\qquad F(-\infty)=F_-(M_0),
\]
and the downstream value $F_-(M_0)$ is independent of $m$ and of the relaxation law.  The figure therefore uses horizontal dashed lines to mark the common end states.  The shock thicknesses reported in the legend quantify the effect of $m$; their global dependence on $M_0$ is shown in Fig.~\ref{fig:shock_thickness_3panels_final}.

Figure~\ref{fig:shock_thickness_3panels_final} shows the resulting dimensionless thickness $\Delta$ as a function of $M_0$.  The parameter set and numerical procedure are those in \eqref{NumericalParameters} and the preceding paragraph.  The exponential thermal-thinning law is the first-order local form of the Arrhenius relation \eqref{kArrhenius}: expanding $E_a/(R\theta)$ about $\theta_0$ and absorbing the constant factor into $k_0$ gives $\gamma=E_a/(R\theta_0^2)$.

The subdivision of the figure into three distinct panels highlights the thermomechanical competition between the weak-shock limit $M_0 \to 1^+$, where the temperature variation is small, and the progressively hotter state induced by viscous dissipation as the upper critical boundary $\Mst$ is approached.  The vertical dotted lines mark the two critical Mach numbers corresponding to $\bmuo=0.10$ and $\bmuo=0.20$.  The curves have been drawn up to the theoretical endpoint limits, so that in the middle and right panels the collapse $\Delta\to0$ at $\Mst$ is shown explicitly.

\begin{figure}[t]
\centering
\includegraphics[width=\columnwidth]{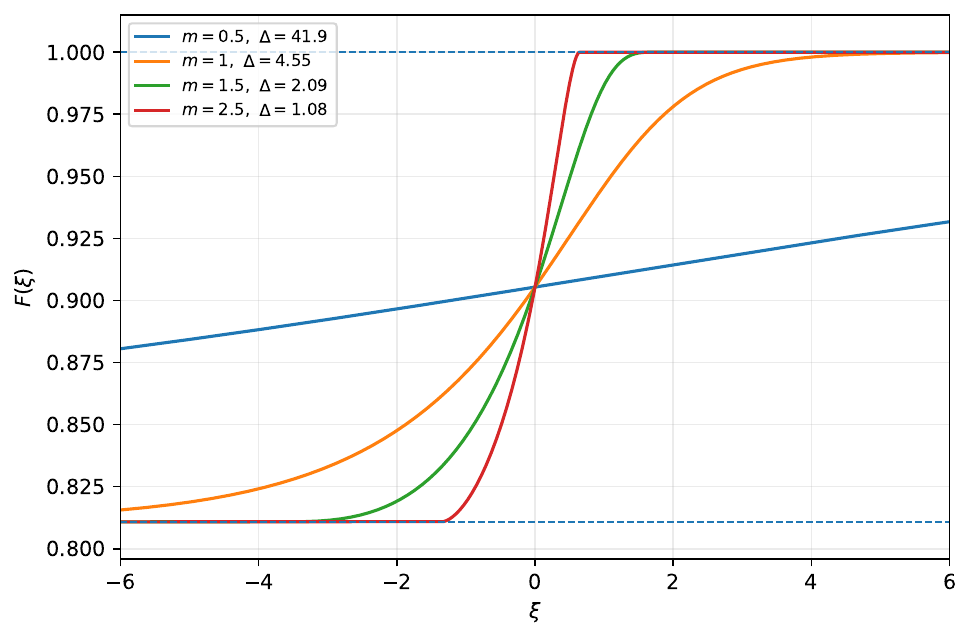}
\caption{Representative nonisothermal RET shock profiles for $M_0=1.4$, $\bmuo=2.0$, and the base parameter set \eqref{NumericalParameters}.  All curves connect the same Rankine--Hugoniot states.  The legend reports the dimensionless maximum-slope thickness $\Delta$.  For both $m=1.5$ and $m=2.5$, the horizontal portions show the $C^1$ constant continuation beyond the finite profile endpoints, as predicted for $m>1$.}
\label{fig:shock_profiles}
\end{figure}

The macroscopic behavior changes drastically depending on the power-law flow index $m$:
\begin{itemize}
\item \textbf{Left panel ($m \le 1$):} Shows that the shock thickness, although diverging in the weak-shock limit for these values of $m$, does not collapse near the critical Mach number. Instead it approaches a finite plateau $\Delta_c>0$, whose value depends on the constitutive parameters used in the numerical calculation.
\item \textbf{Middle panel ($1 < m \le 2$):} Highlights the change in analytical regime for $m>1$. In this range the thickness collapses as the critical boundary $\Mst$ is approached; thermal effects modify the prefactor and the shape of the curves, but not the exponent $(m-1)/m$.
\item \textbf{Right panel ($m > 2$):} Illustrates the regime where the macroscopic thickness vanishes at both ends of the interval of existence of the continuous shock ($M_0\to1^+$ and $M_0\to\Mst$). The zoomed ordinate axis is used only to display the resulting dome-like curves clearly.
\end{itemize}

Finally, Fig.~\ref{fig:thermal_effect} presents a dimensionless sensitivity study that isolates the role of the reference temperature. In this diagnostic only, we use
\[
\rho_0=C_v=K_0=k_0=\tau_0=1,
\quad n=5,
\quad \alpha_T=0.50,
\quad a_T:=\alpha_T K_0=0.50,
\quad \gamma=1.00,
\quad p_\infty=0.15,
\]
and keep the dimensional ratio $\mu_0/(\tau_0\rho_0)=0.30$ fixed while
$0.5\le\theta_0\le1.5$.  Hence
\[
c_S^2(\theta_0)=0.75+0.25\theta_0,
\qquad
\bmuo(\theta_0)=\frac{0.30}{c_S^2(\theta_0)}.
\]
The left panel therefore shows the corresponding downward shift of
$\Mst(\theta_0)$.  The right panel uses $M_0=1.10$ and $m=1.5$; this Mach
number remains subcritical over the whole displayed interval.  The stronger
thermal parameters are used only to make the diagnostic trend visible and
are not used in the analytical results or in Figs.~\ref{fig:shock_profiles}
and \ref{fig:shock_thickness_3panels_final}. They nevertheless remain inside
the analytically sufficient Hugoniot-convexity range
\eqref{TaitConvexCondition}, since
$a_T/(\rho_0C_v)=0.50<n+1=6$.

\begin{figure*}[tb]
\centering
\includegraphics[width=\textwidth]{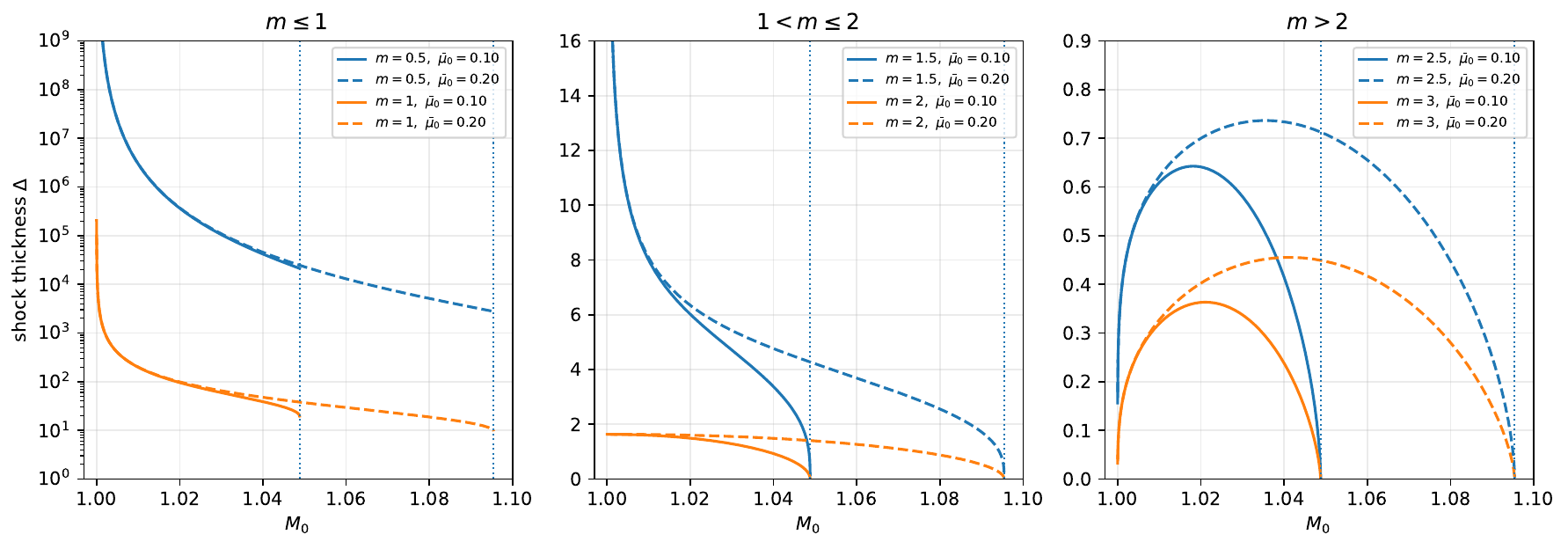}
\caption{Dimensionless thickness $\Delta$ versus $M_0$. Solid and dashed
curves correspond to $\bmuo=0.10$ and $0.20$, and vertical dotted lines mark
$\Mst$. The three panels show respectively the regimes $m\le1$,
$1<m\le2$, and $m>2$, including the predicted endpoint limits. The
thresholds are $m=2$ for weak shocks and $m=1$ near $\Mst$; plateau values
and prefactors are constitutive-dependent.}
\label{fig:shock_thickness_3panels_final}
\end{figure*}

\begin{figure*}[tb]
\centering
\includegraphics[width=0.84\textwidth]{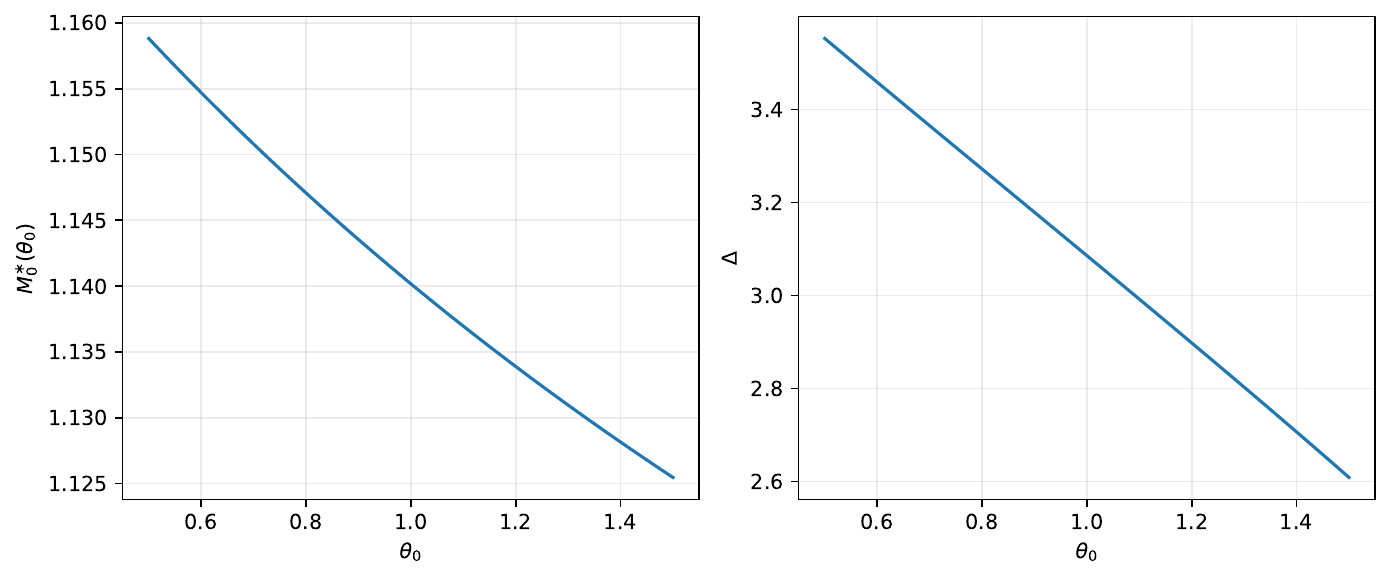}
\caption{Dimensionless sensitivity study of the temperature effect. Left: $\Mst$ when
$\mu_0/(\tau_0\rho_0)=0.30$ is fixed. Right: $\Delta$ at $M_0=1.10$ and
$m=1.5$ for $k(\Theta)=\exp[-(\Theta-\theta_0)]$; the displayed branch is
subcritical throughout.}
\label{fig:thermal_effect}
\end{figure*}

\section{Conclusions}
\label{sec:conclusions}

We have studied the shock-wave structure of a nonisothermal hyperbolic power-law fluid within Rational Extended Thermodynamics. Relative to the published isothermal companion study~\cite{RuggeriTaniguchi}, restoring the energy balance changes the mathematical profile problem, not merely its coefficients. The temperature is selected by an algebraic energy constraint and feeds back into the composite pressure, relaxation production, characteristic denominator, and critical Mach number through the upstream adiabatic sound speed. Corollary~\ref{cor:no-constant-temperature} makes this distinction explicit: under the present assumptions, no nontrivial compressive profile of the full system can remain at constant temperature.

The principal new analytical result is the global characteristic-ordering identity of Lemma~\ref{lem:Dpositive}. It proves that the upstream state is the strict global minimum of the positive nonequilibrium characteristic speed and thereby establishes the complete subcritical continuous branch without an additional regularity hypothesis. Once this genuinely nonisothermal obstruction is resolved, the two scaling mechanisms can be analyzed and are found to be robust: the flow index $m$ determines the constitutive-independent thresholds $m=2$ in the weak-shock limit and $m=1$ in the near-subshock limit, while the finite limits and amplitudes carry the thermodynamic and rheological dependence.

The isothermal theory analyzed with Taniguchi is not an ad hoc truncation. By fixing the thermal component of the main field and omitting the associated energy balance, it is obtained as a principal subsystem of the present nonisothermal equations~\cite{ArimaRuggeri2026,BoillatRuggeri1997}. Hence the entropy structure, symmetric hyperbolicity, and subcharacteristic ordering are retained. This exact relation at the level of the field equations should not, however, be confused with a profile-convergence theorem. Proving that a suitable family of nonisothermal shock profiles converges to an isothermal profile would require the identification of an appropriate thermal limiting parameter and uniform estimates for the traveling-wave problem. Such convergence remains a separate open problem; the result established here is instead the robustness of the two critical thresholds and their exponents under thermomechanical coupling.

The weak-shock exponent is $1-2/m$. Hence the thickness diverges for $m<2$, tends to a finite non-zero limit for $m=2$, and collapses for $m>2$. The threshold and exponent are independent of the particular equation of state, but the finite plateau at $m=2$ and the prefactor of the scaling law depend on the constitutive functions. This dependence enters through the local curvature of the composite pressure and through the temperature profile determined by the energy integral.

The Boillat--Ruggeri theorem~\cite{Breakdown} rules out a $C^1$ shock structure once $M_0>\Mst$; consequently, any admissible piecewise-smooth connection must contain a subshock. The converse subcritical statement is established here for a broad constitutive class: if $p_\theta\ge0$ and the reduced Hugoniot pressure is strictly convex, the positive nonequilibrium characteristic speed has its strict global minimum at the upstream state. Therefore no interior encounter can precede $M_0=\Mst$, and the terminal encounter is automatically transverse under the stated strict-convexity assumptions. In the near-subshock limit, the local degeneracy of the traveling-wave denominator produces the universal maximization structure
$u^{1/m}/(\dMach\delta+\dDef u)$, with $u=1-F$ and $\delta=\Mst-M_0$. The coefficients $\dMach$ and $\dDef$ are defined in \eqref{dMdFdef}, with the explicit constitutive representation of $\dDef$ given in \eqref{dFexplicit}. This yields the threshold $m=1$ and, for $m>1$, the collapse law $\Delta\sim C_c(\Mst-M_0)^{(m-1)/m}$. The exponent is universal, but the coefficient $C_c$ is not. For $0<m\le1$ the thickness $\Delta$ tends to a finite positive value; this value is generally dependent on the equation of state, the internal energy, the temperature profile, and the bounded function $k(\theta)$.

The numerical calculations with a Tait--Murnaghan-type equation of state and a temperature-dependent consistency coefficient illustrate these conclusions for the dimensionless thickness $\Delta$ defined in \eqref{DeltaDef}--\eqref{DeltaPhysicalRelation}. They also make clear the role of the reference temperature $\theta_0$. Through the upstream adiabatic sound speed $c_S(\theta_0)$, the reference temperature changes the dimensionless viscosity parameter $\bmuo$ and hence $\Mst(\theta_0)$ according to~\eqref{bmu0def} and \eqref{criticalMach}. In the Tait--Murnaghan example, when the dimensional viscous--relaxation scale is kept fixed, $c_S(\theta_0)$ increases with $\theta_0$; therefore $\bmuo$ decreases and the critical Mach number moves to lower values. The temperature profile $\Theta(F)$ also affects the shock thickness through the temperature-dependent consistency coefficient $k(\Theta)$. For the thermally thinning law used in the diagnostic calculation, higher temperatures reduce the effective consistency and the computed shock layer becomes thinner. These thermal effects change the numerical thickness, the finite plateaus, and the location of the subshock threshold, but, within the present class of hyperbolic power-law relaxation closures, they do not change the two flow-index thresholds or their critical exponents.

The collapse of the resolved smooth thickness for large $m$ should be interpreted with care. In the present macroscopic RET closure it signals a loss of resolved smooth shock-layer structure in the corresponding asymptotic limit. It should not be identified directly with microstructural jamming or discontinuous shear thickening in dense suspensions, which require additional variables such as particle volume fraction, frictional contacts, dilation, particle migration, and normal-stress effects. The present theory provides a thermodynamically consistent macroscopic classification of shock-thickness regimes; more detailed microstructural models would be needed for a quantitative description of dense-suspension impact.

A broader extension suggested by the analysis is that the two critical
mechanisms should depend primarily on the order with which the production
term vanishes near equilibrium, rather than on its exact power-law form.
Establishing this statement for general entropy-compatible production terms
would require separate hypotheses on the entropy structure,
subcharacteristic condition, and reduced-profile monotonicity, and is left
for future work.

\section*{Acknowledgments}
The work of T.R. was carried out within the activities of GNFM--INdAM.

\subsection*{Funding}
This research received no specific grant from any funding agency in the public, commercial, or not-for-profit sectors.

\subsection*{Conflict of Interest}
The author declares no conflicts of interest.

\subsection*{Data Availability}
The numerical data underlying the figures and the computational code used to generate them are available from the author upon reasonable request.

\end{document}